\documentclass[final,12pt,times]{elsarticle}
\usepackage[left=2cm,right=2cm,top=2cm,bottom=3cm]{geometry}
\usepackage{amssymb}
\usepackage{amsthm}
\usepackage{latexsym}
\usepackage{amsmath}
\usepackage{color}
\usepackage{graphicx}
\usepackage{indentfirst}
\usepackage{mathrsfs}
\usepackage{pdfsync}
\usepackage{hyperref}
\hypersetup{hypertex=true,
	colorlinks=true,
	linkcolor=blue,
	anchorcolor=blue,
	citecolor=blue}
\allowdisplaybreaks
\usepackage[utf8]{inputenc}

\newtheorem{theorem}{\color{black}\indent Theorem}

\newtheorem{lemma}{\color{black}\indent Lemma}[section]
\newtheorem{proposition}{\color{black}\indent Proposition}
\newtheorem{definition}{\color{black}\indent Definition}[section]
\newtheorem{remark}{\color{black}\indent Remark}[section]

\begin{document}
\begin{frontmatter}
\title{ Quantum Birkhoff Normal Form in the $\sigma-$Bruno-R\"{u}ssmann non-resonant condition}

\author{{ Huanhuan Yuan$^{a,}$ \footnote{ E-mail address : yuanhh128@nenu.edu.cn}
        ,~ Yixian Gao$^{a,c,}$} \footnote{E-mail address : gaoyx643@nenu.edu.cn}
		,~ Yong Li$^{b,c,}$  \footnote{E-mail address : liyong@jlu.edu.cn} \\
	{$^{a}$School of Mathematics and Statistics, Northeast Normal University,} {Changchun 130024, P. R. China.}\\
    {$^{b}$School of Mathematics, Jilin University, Changchun 130012, P. R. China.} \\
	{$^{c}$Center for Mathematics and Interdisciplinary Sciences, Northeast Normal University,}
	{Changchun, 130024, P. R. China.}
}

\begin{abstract}
 This paper constructs a quantum Birkhoff normal form in the Gevrey category for an $h$-pseudodifferential operator $P_h(t)$, where $t \in (-\tfrac{1}{2}, \tfrac{1}{2})$, in a neighborhood of a union $\Lambda$ of KAM tori. The construction proceeds from a classical Birkhoff normal form for the principal symbol and is established under a $\sigma$-Bruno--R\"{u}ssmann non-resonance condition with $\sigma > 1$.

\end{abstract}

\end{frontmatter}
\section{Introduction}
Birkhoff normal form is a fundamental tool in the study of Hamiltonian systems and plays a central role in classical KAM theory. For an invariant Lagrangian torus with a non-resonant frequency vector, the Birkhoff normal form provides a systematic procedure to eliminate non-essential terms through successive canonical transformations. As a result, the Hamiltonian can be reduced to a nearly integrable form, revealing the fine dynamical structure of the system in a neighborhood of the torus.
The central role of the Birkhoff normal form in classical KAM theory has been fully demonstrated in the foundational works of Kolmogorov, Arnold, and Moser \cite{MR68687},\cite{MR163025} and \cite{MR147741}.

In the classical setting, a further essential ingredient is the non-resonance condition on the frequency vector, most notably the Diophantine condition. This arithmetic assumption plays a fundamental role in both Birkhoff normal form constructions and KAM theory, as it provides effective control of small divisors. As a consequence, the normalization procedure can be carried out in a convergent manner, leading in particular to the persistence of invariant tori under sufficiently small perturbations.

In a general regularity framework, non-resonance conditions should be adapted to the function space under consideration rather than imposed as fixed assumptions. In the Gevrey setting, where Fourier coefficients exhibit only sub-exponential decay, small divisors accumulate more severely during normalization, rendering the classical Diophantine condition overly restrictive. This naturally leads to the introduction of non-resonance conditions depending on the Gevrey exponent, designed to balance arithmetic control with regularity loss. The $\sigma$-Bruno-R\"{u}ssmann condition considered in this paper follows this principle, providing effective control of small divisors within the Gevrey framework while remaining sufficiently weak. In the analytic limit, it reduces to the classical Bruno-R\"{u}ssmann condition introduced by Bruno \cite{MR377192} and developed by R\"{u}ssmann \cite{MR1843664} in KAM theory.

The construction of Birkhoff normal forms is a cornerstone of classical KAM theory, providing a precise description of Hamiltonian dynamics in a neighborhood of invariant tori. In the semiclassical regime, its quantum counterpart plays an equally fundamental role in the analysis of spectral asymptotics and the microlocal behavior of eigenfunctions. From this perspective, quantum Birkhoff normal forms may be viewed as a semiclassical analogue of the classical normal form, aiming to capture the persistence of local integrable structures after quantization. A natural question then arises: to what extent can the integrable structure revealed by the classical Birkhoff normal form be preserved at the quantum level? More precisely, if the principal symbol of a semiclassical pseudodifferential operator exhibits an integrable structure near an invariant torus, can the corresponding quantum system be simplified, via suitable unitary transformations, to an effective quantum normal form that faithfully reflects its spectral properties and quantum dynamics?

\subsection{Quantum Birkhoff Normal Forms: Motivation and Background}
In classical Hamiltonian systems, the stability of integrable structures and invariant tori is central to understanding dynamical behavior. At the quantum level, a natural question arises: to what extent are these locally integrable properties, revealed by the classical Birkhoff method, preserved after quantization? A series of related semiclassical phenomena, such as the construction of quasi-modes, the clustering and splitting of energy levels, and the localization of eigenfunctions in phase space, all indicate the need for a quantum-level tool to characterize the influence of classical integrable structures on quantum systems.

The quantum Birkhoff normal form was introduced in this context. Its basic idea is to simplify the original quantum Hamiltonian operator near the invariant torus into a "quantum normal form operator" through appropriate unitary transformations (usually generated by quasi-differential operators or Fourier integral operators). The dominant sign of this normal form operator depends only on the action variable, thus corresponding to the integrable Hamiltonian in the classical Birkhoff normal form. Using the quantum Birkhoff normal form, the spectral asymptotic properties of quantum systems, the microlocal distribution of eigenfunctions, and the stability and clustering structure of energy levels under small perturbations can be effectively analyzed. Therefore, it constitutes an important bridge connecting classical integrability and quantum dynamics.

The quantum analogue of the Birkhoff normal form was introduced as a powerful tool in semiclassical analysis to investigate spectral properties and the microlocal structure of eigenfunctions near invariant tori. Popov's \cite{MR1770800} landmark work saw him systematically establish quantum Birkhoff normal form theory for the first time within a quantum framework, under rigorous constraints of the Diophantine condition and Gevrey regularity. This laid a rigorous analytical foundation for quantum KAM theory, proving that regularity can be well preserved at the quantum level. Related developments and refinements have since appeared in various contexts, such as the study of spectral asymptotics, quasi-modes, and quantum localization phenomena near invariant tori (see for example, \cite{MR0501196,MR3743700},\cite{MR1724855}). These works demonstrate that quantum Birkhoff normal forms provide an effective bridge between classical integrable structures and their quantum manifestations.

Recent developments in quantum Birkhoff normal form (QBNF) theory have underscored its central role in semiclassical analysis, particularly in the study of quantum systems near invariant tori under non-resonant conditions. As a quantum counterpart of the classical Birkhoff normal form, the QBNF provides a systematic procedure to simplify semiclassical operators via unitary conjugations, thereby revealing the underlying integrable structure at the quantum level. These advances, relying on techniques from pseudodifferential and Toeplitz operator theory, have significantly deepened our understanding of semiclassical spectral asymptotics and the microlocal behavior of quantum dynamics.

Effective constructions of QBNF have led to precise semiclassical spectral asymptotics \cite{MR2463493} and have been successfully applied to the analysis of discrete spectra near nondegenerate potential wells, yielding uniform energy estimates \cite{MR2423760}. More generally, quantum normal form techniques provide a robust framework for reducing complex quantum Hamiltonians to model operators that depend only on action variables, thereby facilitating a detailed analysis of their spectral properties. Related applications include the microlocal factorization of sub-Riemannian Laplacians by Colin de Verdi`{e}re et al. \cite{MR3743700}, as well as the construction of Birkhoff normal forms for elliptic Fourier integral operators in the semiclassical regime \cite{MR1914459}. For general background on semiclassical analysis and normal form methods, we refer to \cite{MR2952218}.

Beyond these developments, quantum Birkhoff normal forms have also been investigated in various specific models. Imekraz \cite{MR2860609} established QBNFs for semilinear quantum harmonic oscillators under analytic regularity and classical non-resonance assumptions. In the classical setting, related normal form techniques have been employed by Liang and Li \cite{MR4853426} to analyze resonance structures and twist coefficients in nonlinear oscillatory systems. These works collectively highlight the versatility of normal form methods and further motivate a systematic study of QBNFs under weaker regularity and non-resonance conditions.

Motivated by these considerations, we introduce in the next subsection the class of semiclassical pseudodifferential operators under consideration, together with the general assumptions required for the construction of the quantum Birkhoff normal form.

\subsection{General assumptions}

The primary objective of this paper is to construct the quantum Birkhoff normal form (QBNF) for the $h$-differential operator $P_h(t)$ (defined as \eqref{bg}). To that end, we first introduce the analytic and geometric setup of the problem.

Let $M$ be either a compact real analytic manifold of dimension $n \ge 2$ or $\mathbb{R}^n$. We denote by $\Omega^{\frac{1}{2}}$ the bundle of half-densities on $M$, which is the square root of the density bundle $\Omega$. A smooth section $u \in C^\infty(M, \Omega^{\frac{1}{2}})$ can be locally expressed as $u(x) |\mathrm{d}x|^{\frac{1}{2}}$, where $u \in C^\infty(M)$ and $|\mathrm{d}x|^{\frac{1}{2}}$ transforms under coordinate changes by the square root of the absolute value of the Jacobian determinant. The space of half-densities is particularly natural when considering formally self-adjoint differential or pseudodifferential operators, as it allows for coordinate-invariant integration without the need to specify a Riemannian metric.

We consider a formally self-adjoint $h$-pseudodifferential operator $P_{h}(t)$ of finite order $m$, acting on half-densities in $C^{\infty}(M,\Omega^{\frac{1}{2}})$, of the form
\begin{align}\label{bg}
P_{h}(t):=P_{h}(x,hD;t)=\sum_{j=0}^{m}P_{j}(x,hD;t)h^{j}.
\end{align}
Here, each $P_{j}(x,\xi;t)$ is a polynomial in $\xi$ with analytic coefficients, and $D=(D_{j})_{j=1}^{n}$, where $D_{j}=-{\rm i}\frac{\partial}{\partial x_{j}}$. In particular, for $j=0$, $P_{0}(x,\xi;t)$ represents the principal symbol of $P_{h}(t)$, denoted by $H(x,\xi;t)=P_{0}(x,\xi;t)$, for $(x,\xi;t)\in T^{*}M \times(-\frac{1}{2},\frac{1}{2})$. For $j=1$, $P_{1}(x,\xi;t)$ is the subprincipal symbol, although we generally set $P_{1}(x,\xi;t)=0$.

A canonical example is the semiclassical Schr\"{o}dinger operator
\begin{align}
P_{h}=-h^{2}\triangle +V(x),
\end{align}
where $\triangle$ is the Laplace--Beltrami operator on $M$, and $V(x)$ is a real analytic potential on $M$ bounded from below.
It is then straightforward to verify that the symbol of $P_{h}$ is $p(x,\xi)=|\xi|^{2}+V(x)$ for $x\in M$, which is elliptic.

In this work, we aim to construct a Gevrey QBNF for $P_{h}(t)$, with $t\in(-\frac{1}{2},\frac{1}{2})$,
around the union $\Lambda$ of KAM tori.
This construction starts from a suitable Birkhoff normal form of $H$ near $\Lambda$, under a generalized non-resonance condition known as the $\sigma$-Bruno-R\"{u}ssmann condition (see Definition \ref{i}).

The main contribution of this paper is to construct quantum Birkhoff normal form construction to a class of weakly non-resonant frequencies satisfying a $\sigma$-Bruno-R\"{u}ssmann condition. More precisely:
\begin{itemize}
  \item We introduce a $\sigma$-Bruno-R\"{u}ssmann type arithmetic condition allowing logarithmic control of small divisors. See section \ref{sec2.2} for details.
  \item We prove that the quantum Birkhoff normal form can still be constructed for Gevrey regular symbols, without loss of Gevrey regularity, despite the weaker arithmetic assumptions.
  \item The analysis is carried out in the presence of an external parameter $t$, yielding uniform estimates with respect to this parameter.
\end{itemize}

\subsection{Main result }

This work establishes a quantum Birkhoff normal form (QBNF) for the family of semiclassical operators $P_h(t)$ in a neighborhood of the KAM invariant tori $\Lambda$. Our construction builds upon the classical Birkhoff normal form of the Hamiltonian $H$ near $\Lambda$, which we now briefly review.

For a Hamiltonian $H\in \mathcal{G}^{\rho,\rho,1}(\mathbb{T}^{n}\times D\times(-\frac{1}{2},\frac{1}{2}))$ (see Definition \eqref{a12}), assume that there exists a real analytic exact symplectic diffeomorphism
\begin{align*}
\chi_{t}^{1}:\mathbb{T}^{n}\times D\rightarrow U\subset T^{*}M,
\end{align*}
where $D$ is an open subdomain in $\mathbb{R}^{n}$, such that the transformed Hamiltonian
$
\hat{H}(\varphi,I):=(H\circ\chi^{1}_{t})(\varphi,I)
$
admits a $\mathcal{G}^{\rho,\rho+1,\rho+1}$-Birkhoff normal form around a family of invariant tori with frequencies in a suitable $\Omega_{\sigma}$.
That is, there exists an exact symplectic transformation $\chi_{t}^{0}\in \mathcal{G}^{\rho,\rho+1}(\mathbb{T}^{n}\times D,\mathbb{T}^{n}\times D)$ with a generating function $\Phi\in\mathcal{G}^{\rho,\rho+1,\rho+1}(\mathbb{T}^{n}\times D\times(-\frac{1}{2},\frac{1}{2}))$, such that
\[\hat{H}(\chi^{0}_{t}(\varphi,I))=K_{0}(I;t)+R_{0}(\varphi,I;t), \quad \text{in}~ \mathbb{T}^{n}\times D\times\left(-\frac{1}{2},\frac{1}{2}\right),\]
where $K_{0}(I;t)\in\mathcal{G}^{\rho+1,\rho+1}(D\times(-\frac{1}{2},\frac{1}{2}))$ (see Definition \eqref{at}) and $R_{0}\in \mathcal{G}^{\rho,\rho+1,\rho+1}(\mathbb{T}^{n}\times D\times(-\frac{1}{2},\frac{1}{2}))$ satisfy $D_{I}^{\alpha}R_{0}(\varphi,I;t)=0$ and $D_{I}^{\alpha}(\nabla K_{0}(I;t)-\omega(I;t))=0$ for any $(\varphi,I;t)\in \mathbb{T}^{n}\times \omega^{-1}(\Omega_{\sigma};t)\times(-\frac{1}{2},\frac{1}{2})$. Additionally, the generating function $\Phi\in\mathcal{G}^{\rho,\rho+1,\rho+1}(\mathbb{T}^{n}\times D\times(-\frac{1}{2},\frac{1}{2}))$ for $\chi^{0}_{t}$ satisfies
\[\|Id -\Phi_{\theta I}(\varphi,I;t)\|\leq \varepsilon \quad \text{in}~ \mathbb{T}^{n}\times D\times\left(-\frac{1}{2},\frac{1}{2}\right)\]
for some $0<\varepsilon<1$, as stated in \cite{MR1770799} for the case without the variable $t$. Here, we denote by $\Phi_{\theta I}$ the matrix of second-order partial derivatives of $\Phi$ with respect to $\theta$ and $I$.

We start with the classical Birkhoff normal form of the Hamiltonian $H$, which is conjugated to a normal form of the class $\mathcal{G}^{\rho,\rho+1,\rho+1}$ through Gevrey symplectic transformations in the neighborhood of $\Lambda$.
On this basis, we construct a family of unitary Fourier integral operators $U_{h}(t)$, which are microlocally related to the classical transformation $\chi_{t}=\chi^{0}_{t}\circ\chi^{1}_{t}$, and whose graphs define the canonical relation of $U_{h}(t)$ (see Definition~\ref{bu}).
Here, $U_{h}(t)$ is constructed based on the quantization of $\chi_{t}$.
The conjugate operator $U^{*}_{h}(t)\circ P_{h}(t)\circ U_{h}(t)$ has quantum Birkhoff normal form in the class of $h$-pseudodifferential operators on $L^{2}(\mathbb{T}^{n},\mathbb{L})$, where $\mathbb{L}$ denotes the flat Hermitian line bundle over $\mathbb{T}^{n}$. Its full symbol lies in $p(\varphi,I;t,h) \in \mathcal{S}_{l}(\mathbb{T}^{n}\times D\times (-\frac{1}{2},\frac{1}{2}))$, where $l=(\sigma,\mu,\lambda,\bar{\rho})$ (see Definition~\ref{bv}).

The self-adjoint $h$-differential operator $P_{h}(t)$ given by \eqref{bg} acts on half densities in $C^{\infty}(M,\Omega^{\frac{1}{2}})$ associated with the principal symbol $H$ and vanishing subprincipal symbol. We assert that Theorem~\ref{ac} holds under the $\sigma$-Bruno-R\"{u}ssmann condition, as shown below:

\begin{theorem}\label{ac}

Let $H \in \mathcal{G}^{\rho,\rho,1}$ be a Hamiltonian. Suppose there exists a real analytic symplectic transformation
$\chi^{1}_{t} : \mathbb{T}^{n}\times D \rightarrow U \subset T^{*}M$ such that the transformed Hamiltonian
$\hat{H} = H(\chi^{1}_{t}(\varphi,I))$, for $(\varphi,I) \in \mathbb{T}^{n}\times D$, admits a Birkhoff normal form via
an exact symplectic transformation $\chi_{t}^{0}$. Then there exists a family of unitary $h$-Fourier integral operators
\[
U_h(t) : L^2(\mathbb{T}^n; \mathbb{L}) \rightarrow L^2(M), \quad \text{for } 0 < h \leq h_0,
\]
associated with the canonical relation graph, such that:
\[
P_h(t) \circ U_h(t) = U_h(t) \circ P_h^0(t),
\]
where $P_h^0(t)$ is an $h$-pseudodifferential operator with full symbol
\[
p^0(\varphi,I;t,h) = K^0(I;t,h) + R^0(\varphi,I;t,h),
\]
and
\[
K^0(I;t,h) = \sum_{0 \leq j \leq \eta h^{-1/\rho}} K_j(I;t) h^j, \quad
R^0(\varphi,I;t,h) = \sum_{0 \leq j \leq \eta h^{-1/\rho}} R_j(\varphi,I;t) h^j
\]
for some constant $\eta > 0$. Both $K^0$ and $R^0$ belong to the Gevrey class
$\mathcal{S}_l(T^{*}\mathbb{T}^n\times (-\frac{1}{2},\frac{1}{2}))$ with $K^0$ real-valued. Moreover,
$R^0$ vanishes to infinite order on the Cantor set $\mathbb{T}^n\times E_{\kappa}(t)$, where
$E_{\kappa}(t) = \omega^{-1}(\Omega_{\sigma};t)$ denotes the set of nonresonant actions for
$t \in (-\frac{1}{2},\frac{1}{2})$.
\end{theorem}

\begin{remark}
Although the resulting quantum Birkhoff normal form is formally analogous to that obtained under Diophantine conditions, the underlying proof strategy is fundamentally different. In the present setting, the control of small divisors relies crucially on the stretched sub-exponential decay associated with the $\sigma$-Bruno-R\"{u}ssmann condition, rather than on polynomial Diophantine bounds.
\end{remark}

Building on Popov's seminal work \cite{MR1770800}, we consider a natural extension of the quantum Birkhoff normal form framework by replacing the Diophantine non-resonance assumption with a class of $\sigma$-approximation functions satisfying suitable monotonicity and integrability conditions. Our analysis is carried out in a broader Gevrey regularity class $\mathcal{G}^{\rho,\rho,1}$ and allows for analytic dependence on a time parameter $t$. Moreover, by exploiting the subexponential decay of Fourier coefficients inherent to the Gevrey setting, our approach avoids truncation and provides direct control of small divisors, enabling the construction of the quantum Birkhoff normal form under non-resonance conditions weaker than the Diophantine case.

Once a Gevrey quantum Birkhoff normal form is established, it is expected to provide a natural starting point for the construction of Gevrey quasimodes and for further spectral analysis near invariant tori, as in earlier works under stronger non-resonance assumptions. Interested readers can refer to the reference \cite{MR333485} (see also \cite{MR1239173}), \cite{MR0501196},\cite{MR3746630} and \cite{MR3742467}.

\subsection{Outline of the paper}

This paper is organized as follows.

In Section \ref{sec2}, we define the classes of $h$-pseudodifferential and $h$-Fourier integral operators employed throughout, along with the relevant Gevrey symbol classes. We introduce a new family of non-resonance conditions based on $\sigma$-approximation functions, generalizing classical Bruno-R\"{u}ssmann conditions. A key technical tool is the supremum function
\[
\Gamma_s(\eta) \triangleq \sup_{Q \geq 0} \, (1+Q)^s \Delta(Q) e^{-\eta Q^{\frac{1}{\sigma}}},
\]
which remains bounded under our assumptions. This illustrates a fundamental mechanism: the subexponential decay inherent to Gevrey regularity compensates for losses induced by weak non-resonance conditions, enabling control of small divisors and solution of the cohomological equation without truncation.

Section \ref{sec3} constructs the microlocally unitary $h$-FIO $U_h(t)$ associated with the classical Birkhoff normal form transformations by quantizing the canonical maps $\chi_t^1$ and $\chi_t^0$.

Section \ref{sec4} is devoted to the proof of Theorem \ref{ac}, establishing the quantum Birkhoff normal form for $P_h(t)$ near a Gevrey family of KAM tori. The construction involves conjugating $P_h(t)$ to an $h$-pseudodifferential operator $\widetilde{P}_h(t)$ whose principal symbol is in classical Birkhoff normal form, and solving a cohomological equation over a Cantor set of non-resonant actions using Gevrey Whitney extensions.

Finally, Section \ref{sec5} completes the proof of Theorem \ref{ad}, building upon the conjugation procedures and Fourier estimates developed in preceding sections.

\subsection{Related literature}

In quantum mechanics, the evolution of physical systems is governed by the Schr\"{o}d-\\inger equation, which requires analytical tools fundamentally different from those of classical dynamics.
A significant body of work on quantum Birkhoff normal forms has been developed by Popov (see \cite{MR2111816}--\cite{MR1770800}). Employing Gevrey-class techniques, he constructed a quantum Birkhoff normal form (QBNF) for a class of semiclassical differential operators $P_{h}$ with principal symbol $p(x,\xi)$ and vanishing subprincipal symbol. This construction produced quasimodes localized near nonresonant invariant tori, with errors that are exponentially small in the semiclassical parameter $h$ (see \cite{MR2111816}).

More recently, Gomes has made notable contributions to the study of quantum scarring and the failure of quantum ergodicity. In \cite{MR4404789}, he showed that on smooth, compact surfaces, for a generic one-parameter family of Hamiltonians and almost all KAM tori $\Lambda_{\omega}$ (with Diophantine frequency $\omega$), there exists a semiclassical measure with positive mass supported on these invariant tori for almost every $t\in(0,\delta)$. In a subsequent work \cite{MR4578527}, he further proved that, under generic conditions, the quantization of such one-parameter families fails to be quantum ergodic for a full-measure set of parameters $t\in(0,\delta)$.

\section{Preliminaries} \label{sec2}
\subsection{h-PDOs, h-FIOs,  Gevrey symbol classes}

Pseudodifferential operators furnish a systematic framework that unifies the treatment of differential and integral operators, fundamentally grounded in the Fourier transform \(\mathcal{F}\) and its inverse \(\mathcal{F}^{-1} = \mathcal{F}^*\). These linear operators are represented by their symbols---functions on phase space that act as generalized Fourier multipliers. The resulting class of pseudodifferential operators forms an algebra, wherein operations such as composition, transpose, and adjoint can be analyzed through algebraic and asymptotic computations on the corresponding symbols. This symbolic calculus reveals a deep unity between differential and integral operations, providing a powerful perspective in linear analysis. Consequently, pseudodifferential operators have become an indispensable tool in both pure and applied mathematics, especially in the theory of partial differential equations and harmonic analysis.

We now introduce the relevant class of Gevrey symbols (for further details, see \cite{MR4404789}). Let \(D_{0}\) be a bounded domain in \(\mathbb{R}^{n}\). Fix parameters \(\sigma,\mu,\lambda>1\) and \(\varrho\geq\mu+\sigma+\lambda-1\), and set \(l=(\sigma,\mu,\lambda,\varrho)\). We define a class of formal Gevrey symbols, denoted by \(F\mathcal{S}_{l}(\mathbb{T}^{n}\times D_{0})\), as follows. Consider a sequence of smooth functions \(p_{j}\in C_{0}^{\infty}(\mathbb{T}^{n}\times D_{0})\) for \(j\in\mathbb{Z}_{+}\), each supported in a fixed compact subset of \(\mathbb{T}^{n}\times D_{0}\).

\begin{definition}
A formal series
\begin{align}\label{a1}
\sum_{j=0}^{\infty} p_{j}(\varphi, I; t) h^{j}
\end{align}
is said to belong to the class \(F\mathcal{S}_{l}(\mathbb{T}^{n}\times D_{0})\) if there exists a constant \(C>0\) such that for every multi-index \(\alpha, \beta, \delta \in \mathbb{N}^{n}\) and every integer \(j \ge 0\), the estimate
\begin{align*}
\sup_{(\varphi, I) \in \mathbb{T}^{n} \times D_{0}}
\bigl| \partial_{I}^{\alpha} \partial_{\varphi}^{\beta} \partial_{t}^{\delta}
p_{j}(\varphi, I; t) \bigr|
\leq
C^{\,j + |\alpha| + |\beta| + |\delta| + 1}
\, \alpha!^{\mu} \, \beta!^{\sigma} \, \delta!^{\lambda} \, j!^{\varrho}
\end{align*}
holds uniformly in \(t\) over its domain.
\end{definition}

\begin{definition}
A function \(p(\varphi, I; t, h)\), defined for \((\varphi, I) \in \mathbb{T}^{n} \times \mathbb{R}^{n}\), is called a \emph{realization} of the formal symbol~\eqref{a1} in \(\mathbb{T}^{n} \times D_{0}\) if the following conditions hold:
\begin{enumerate}
    \item For each fixed \(0 < h \leq h_{0}\), the function \(p(\cdot, \cdot; t, h)\) is smooth in \((\varphi, I)\) and has compact support contained in \(\mathbb{T}^{n} \times D_{0}\).

    \item There exists a constant \(C > 0\) such that for all multi-indices \(\alpha, \beta, \delta \in \mathbb{N}^{n}\), all integers \(N \geq 0\), and all \(h \in (0, h_{0}]\),
    \begin{align*}
    \sup_{(\varphi, I) \in \mathbb{T}^{n} \times D_{0}}
    \Bigl| \partial_{I}^{\alpha} \partial_{\varphi}^{\beta} \partial_{t}^{\delta}
    &\Bigl( p(\varphi, I; t, h) - \sum_{j=0}^{N} p_{j}(\varphi, I) h^{j} \Bigr) \Bigr| \\
    & \leq h^{N+1} \,
    C^{\, N + |\alpha| + |\beta| + |\delta| + 2}
    \, \beta!^{\sigma} \, \alpha!^{\mu} \, \delta!^{\lambda} \, (N+1)!^{\varrho}.
    \end{align*}
\end{enumerate}
\end{definition}

For example,  when the dependence on the variable $t$ is omitted, a realization can be constructed via the truncated series
\[
p(\varphi, I; h)=\sum_{j \leq \varepsilon h^{-1/\bar{\rho}}} p_{j}(\varphi, I)\, h^{j},
\]
where the constant $\varepsilon>0$ depends only on the dimension $n$ and the constant $C$ appearing in the symbol estimates (for the classical case $\sigma=\mu=\varrho=1$, see~\cite{MR1188076}, Sect.1).
We denote by $\mathcal{S}_{l}(\mathbb{T}^{n}\times D_{0})$ the resulting class of symbols, where now
\( l=(\sigma,\mu,\varrho) \) and \(\varrho>\sigma+\mu-1\).

Given two symbols \(p,q\in\mathcal{S}_{l}(\mathbb{T}^{n}\times D_{0})\), their composition
\(p\circ q\in\mathcal{S}_{l}(\mathbb{T}^{n}\times D_{0})\) is defined as the realization of the formal symbol
\[
\sum_{j=0}^{\infty}c_{j}h^{j}\;\in\;F\mathcal{S}_{l}(\mathbb{T}^{n}\times D_{0}),
\]
where the coefficients are given by
\[
c_{j}(\varphi, I; t)=\sum_{\substack{r,s\ge0,\;\gamma\in\mathbb{N}^{n}\\ r+s+|\gamma|=j}}
\frac{1}{\gamma!}\;
\partial_{I}^{\gamma}p_{r}(\varphi,I;t)\;
\partial_{\varphi}^{\gamma}q_{s}(\varphi,I;t).
\]
For a symbol \(p \in \mathcal{S}_{l}(\mathbb{T}^{n} \times D_{0})\), we define the corresponding \(h\)-pseudodifferential operator ($h$-PDO) as follows.

\begin{definition}\label{bv}
Given a symbol \(p \in \mathcal{S}_{l}(\mathbb{T}^{n} \times D_{0})\), the associated \(h\)-pseudodifferential operator \(P_{h}(t)\) acts on a  function \(u \in C^{\infty}_{0}(\mathbb{T}^{n})\) by
\[
\bigl(P_{h}(t)u\bigr)(x) = (2\pi h)^{-n}
\int_{\mathbb{T}^{n}}\int_{\mathbb{R}^{n}}
e^{{\rm i}(x-y)\cdot\xi/h} \,
p(x,\xi;t,h) \,
u(y) \,
\mathrm{d}\xi\,\mathrm{d}y,
\qquad x \in \mathbb{T}^{n},
\]
where the symbol \(p\) is extended by zero outside \(D_{0}\) in the \(\xi\)-variable.
\end{definition}

\begin{definition}
The \emph{residual symbol class} \(\mathcal{S}_{l}^{-\infty}\) is defined as the set of all symbols
\(p \in \mathcal{S}_{l}(\mathbb{T}^{n} \times D_{0})\) that are realizations of the zero formal symbol,
i.e., the formal symbol \(\sum_{j=0}^{\infty} p_{j} h^{j}\) with \(p_{j} \equiv 0\) for every \(j\).
\end{definition}

\begin{remark}
In the Gevrey class \(\mathcal{S}_{l}\), residual symbols exhibit \emph{exponential decay} in \(h\), which is a much stronger property than the rapid decay  characteristic of residual symbols in the standard Kohn--Nirenberg calculus. This strengthening is a direct consequence of the Gevrey-type estimates and the truncation method used to construct realizations.
\end{remark}

We now define $h$-Fourier integral operators with parameter dependence, which play a crucial role in subsequent sections.
\begin{definition}\label{bu}
Let \(X \subset \mathbb{R}^{n}\) and \(Y \subset \mathbb{R}^{m}\) be open sets, and let
\(\phi(x,y,\xi;t)\) be a real-valued phase function defined on \(X \times Y \times (\mathbb{R}^{N} \setminus \{0\})\) that is positively homogeneous of degree \(1\) in \(\xi\) and smooth in all variables for \(\xi \neq 0\). Given a symbol \(f \in \mathcal{S}_{l}(X \times Y \times \mathbb{R}^{N})\), the associated \(h\)-Fourier integral operator is defined by
\[
\bigl(F_{h}(t)u\bigr)(x) = (2\pi h)^{-n}
\int_{Y}\int_{\mathbb{R}^{N}}
e^{{\rm i}\phi(x,y,\xi;t)/h} \,
f(x,y,\xi;t,h) \,
u(y) \,
\mathrm{d}\xi\,\mathrm{d}y,
\qquad x \in X,
\]
for \(u \in C^{\infty}_{0}(Y)\).
\end{definition}

Note that when the phase function takes the form \(\phi(x,y,\xi;t) = \langle x-y, \xi \rangle\),
the operator \(F_h(t)\) reduces to an \(h\)-pseudodifferential operator.
Thus, the class of Fourier integral operators strictly contains the class of pseudodifferential operators.
For a comprehensive treatment of Fourier integral operators, we refer the reader to
\cite{MR1996773}--\cite{MR4248008} and the references therein.

\subsection{Approximation function}\label{sec2.2}
In his work on small divisor problems \cite{MR0612810}, R\"{u}ssmann introduced the notion of an \emph{approximation function} to control a broad class of small denominators, ensuring the convergence of KAM iterations. Following his terminology, we refer to such functions as $\sigma$-\emph{approximation functions}. Remarkably, R\"{u}ssmann's approach does not rely on iterative techniques but rather on sophisticated optimization methods; for a concise exposition in the case $\sigma=1$, see \cite{MR879908}.

We now extend this concept to the Gevrey setting.

\begin{definition}
[$\sigma$-approximation function]\label{def:sigma-approx}
Let $\sigma>1$. A nondecreasing function $\Delta \colon [0,\infty) \to [1,\infty)$ is called a \emph{$\sigma$-approximation function} if it satisfies the following conditions:
\begin{enumerate}
    \item $\displaystyle \frac{\log \Delta(Q)}{Q^{1/\sigma}} \searrow 0 \quad \text{as } Q \to \infty$,
    \item $\displaystyle \int_0^{\infty} \frac{\log \Delta(Q)}{Q^{1+\frac{1}{\sigma}}}\, \mathrm{d}Q < \infty$.
\end{enumerate}
We also impose the normalization $\Delta(0)=1$.
\end{definition}

Typical examples of $\sigma$-approximation functions include:
\begin{itemize}
    \item $\Delta(Q) = (1+Q)^n$ for $n \ge 1$;
    \item $\Delta(Q) = \exp\!\bigl(\frac{Q^a}{a}\bigr)$ for $0 < a < \frac{1}{\sigma}$;
    \item $\Delta(Q) = \exp\!\Bigl(\frac{Q^{1/\sigma}}{1+\log^{\gamma}(1+Q)}\Bigr)$ for $\gamma > 1$.
\end{itemize}

We establish the primary assumption on the frequency $\omega$.
Given a bounded domain $\Omega\subset\mathbb{R}^{n}$, let $\Omega_{\sigma}\subset\mathbb{R}^{n}$ be a frequency set consisting of strongly non-resonant frequencies. That is, each $\omega\in \Omega_{\sigma}$ satisfies
\begin{align}\label{i}
|\langle k,\omega\rangle|\geq\frac{\kappa}{\Delta(|k|)},\quad 0\neq k\in\mathbb{Z}^{n}, \quad  |k|=|k_{1}|+|k_{2}|+\cdots+|k_{n}|
\end{align}
for some $\kappa>0$ and a $\sigma$-approximation function $\Delta$, which is a continuous, strictly increasing, unbounded function $\Delta:[0,\infty)\rightarrow[1,\infty)$ satisfying
\begin{align}\label{yy}
\frac{\log\Delta(Q)}{Q^{\frac{1}{\sigma}}}\searrow 0,\quad \text{as}~ Q\rightarrow\infty,
\end{align}
and
\begin{align}\label{zz}
\int_{\varsigma}^{\infty}\frac{\log\Delta(Q)}{Q^{1+\frac{1}{\sigma}}} {\rm d}Q<\infty,\quad \sigma>1,
\end{align}
where $\varsigma$ is a positive constant. If $\omega$ satisfies \eqref{i},\eqref{yy}, and \eqref{zz}, we say that $\omega$ satisfies the $\sigma$-Bruno-R\"{u}ssmann condition. It is straightforward to verify that $\Omega_{\sigma}$ shrinks as $\sigma$ increases. For example, if $\Delta(Q)=\exp(\frac{Q^{a}}{a})$, then $\omega\in \Omega_{\sigma}$ if and only if $a<\frac{1}{\sigma}$.

Let $\Omega_{\tau}$ be the set of $\tau$- Diophantine vectors $(\tau>n-1)$,  where $\Delta(|k|)\leq |k|^{\tau}$ for all $0\neq k\in\mathbb{Z}^{n}$. The set  $\Omega_{\tau}$ is  non-empty and has full measure if $\tau>n-1$. By definition, we have $\Omega_{\tau}\subset\Omega_{\sigma}$ (see \cite{MR4050197}). Therefore, $\Omega_{\sigma}$ is non-empty.

To quantify the effect of small divisors encoded by an approximation function \(\Delta\), we introduce an auxiliary function \(\Gamma_s\) derived from \(\Delta\).
Given a \(\sigma\)-approximation function \(\Delta\), define for each \(s \ge 1\) and \(\eta > 0\) the function
\[
\Gamma_s(\eta) = \sup_{Q \ge 0} \; (1+Q)^s \, \Delta(Q) \, e^{-\eta Q^{1/\sigma}} .
\]
Notice that if \(\Delta\) itself is a \(\sigma\)-approximation function, then for any \(s \ge 1\) the function \(Q \mapsto (1+Q)^s \Delta(Q)\) is again a \(\sigma\)-approximation function; in particular, both \((1+Q)^s\) and \((1+Q)^s\Delta(Q)\) satisfy the defining conditions \eqref{yy} and \eqref{zz} for a \(\sigma\)-approximation function.
The supremum in the definition of \(\Gamma_s(\eta)\) is finite for every \(\eta > 0\) because of the decay condition \eqref{yy}.

Motivated by the interplay between arithmetic conditions and regularity, we work in the Gevrey class
\begin{align*}
\mathcal{G}^{\rho,\rho,1}(\mathbb{T}^{n}\times D\times (-\frac{1}{2},\frac{1}{2})),
\end{align*}
and introduce a corresponding class of non-resonance conditions adapted to this regularity framework.

In the analytic setting, the classical Bruno condition is formulated in terms of an approximation function $\Delta(Q)$ satisfying
\begin{align*}
\int_{1}^{\infty}\frac{\log\Delta(Q)}{Q^{2}} {\rm d}Q<\infty,
\end{align*}
which ensures convergence of the normalization procedure by balancing small divisors with exponential Fourier decay.

In our Gevrey context, this condition is naturally generalized to
\begin{align*}
\int_{\varsigma}^{\infty} \frac{\log \Delta(Q)}{Q^{1 + \frac{1}{\sigma}}} \, \mathrm{d}Q < \infty, \quad \sigma > 1.
\end{align*}
This $\sigma$-Bruno-R\"{u}ssmann type condition provides a quantitative arithmetic control precisely matched to the Gevrey regularity, allowing the accumulation of small divisors to be compensated by the available decay.

A key point in our approach is that we do not perform Fourier truncation when solving the cohomological equation \eqref{h}. To ensure solvability, we instead impose a mild monotonicity condition:
\begin{align*}
\frac{\log\Delta(Q)}{Q^{\frac{1}{\sigma}}}\searrow 0,\quad \text {as}~ 0<Q\rightarrow\infty,
\end{align*}
which allows the approximation function $\Delta(Q)$ to grow faster than any polynomial, while still being effectively balanced by the subexponential Fourier decay characteristic of the Gevrey class.
This compensation mechanism is encoded in a crucial estimate on the auxiliary function $\Gamma_{s}(\eta)$,
which remains bounded under our assumptions.
This interplay between Gevrey regularity and weakened frequency conditions is a central contribution of the present work. See Section \ref{sec2.2} for details.

Given \(\eta > 0\), let \(\{\eta_\nu\}_{\nu \ge 0}\) be a nonincreasing sequence of positive numbers such that \(\sum_{\nu\ge0} \eta_\nu \le \eta\).
Working with the supremum in the definition of \(\Gamma_s(\eta)\) along such a sequence yields the following estimate:

\begin{lemma}\label{a7}
Let $1 < \kappa \leq 2$ and $T \geq \varsigma$. If
\[
\frac{1}{\log \kappa} \int_T^{\infty} \frac{\log \Delta(Q)}{Q^{1 + 1/\sigma}}  {\rm d}Q \leq \eta,
\]
then
\[
\Gamma_s(\eta) \leq e^{\eta(s) T^{1/\sigma}},
\]
where $\eta(s)$ denotes a constant depending on $s$ and the given $\eta$.
\end{lemma}

\begin{proof}
Define
\[
\delta_s(Q) = \log\bigl((1+Q)^s \Delta(Q)\bigr) = \log\Delta(Q) + s\log(1+Q),
\]
and for each \(\nu \ge 0\) set
\[
Q_\nu = \kappa^\nu T, \qquad
\eta_\nu = \frac{\delta_s(Q_\nu)}{Q_\nu^{1/\sigma}} .
\]
Since the function \(1+Q\) is itself a \(\sigma\)-approximation function, we may write
\[
\eta_\nu = \wp_\nu + s c_\nu,
\qquad
\wp_\nu = \frac{\log\Delta(Q_\nu)}{Q_\nu^{1/\sigma}}, \quad
c_\nu = \frac{\log(1+Q_\nu)}{Q_\nu^{1/\sigma}} .
\]
By condition \eqref{yy} the sequence \(\{\eta_\nu\}\) is nonincreasing and positive.

Using the integral constraint \eqref{zz}, and the substitution $Q = \kappa^\nu T$, we estimate
\[
\sum_{\nu \geq 0} \wp_\nu \leq \int_0^\infty \frac{\log \Delta(Q_\nu)}{Q_\nu^{1/\sigma}}  {\rm d}\nu
\leq \frac{1}{\log \kappa} \int_T^\infty \frac{\log \Delta(Q)}{Q^{1 + 1/\sigma}}  {\rm d}Q \leq \wp,
\]
and similarly,
\[
\sum_{\nu \geq 0} c_\nu \leq \int_0^\infty \frac{\log(1 + Q_\nu)}{Q_\nu^{1/\sigma}}  {\rm d}\nu
\leq \frac{1}{\log \kappa} \int_T^\infty \frac{\log(1 + Q)}{Q^{1 + 1/\sigma}}  {\rm d}Q \leq c.
\]
Thus,
\[
\sum_{\nu \geq 0} \eta_\nu \leq \int_0^\infty \frac{\delta_s(Q_\nu)}{Q_\nu^{1/\sigma}}  {\rm d}\nu
\leq \frac{1}{\log \kappa} \int_T^\infty \frac{\delta_s(Q)}{Q^{1 + 1/\sigma}}  {\rm d}Q \leq \eta,
\]
where $\eta = \wp + s c$.

Since \(\delta_s(Q)/Q^{1/\sigma}\) is nonincreasing by condition~\eqref{yy}, we have
\(\delta_s(Q) - \eta_\nu Q^{1/\sigma} \le 0\) for all \(Q \ge Q_\nu\).
On the interval \([0, Q_\nu]\), the function \(Q \mapsto \delta_s(Q) - \eta_\nu Q^{1/\sigma}\) attains a maximum, and because \(\delta_s(Q) \le \delta_s(Q_\nu)\) for \(Q \le Q_\nu\), this maximum does not exceed \(\delta_s(Q_\nu)\). Consequently,
\[
\Gamma_s(\eta_\nu) = \sup_{Q \ge 0} \exp\!\bigl(\delta_s(Q) - \eta_\nu Q^{1/\sigma}\bigr)
\le e^{\delta_s(Q_\nu)} = e^{\eta_\nu Q_\nu^{1/\sigma}},
\]
where the last equality uses the definition \(\eta_\nu = \delta_s(Q_\nu)/Q_\nu^{1/\sigma}\).

The function \(\Gamma_s(\eta)\) is monotonically decreasing in \(\eta\).
From the inequality \(\sum_{\nu\ge0} \eta_\nu \le \eta\) we deduce \(\eta_\nu \le \eta\) for every \(\nu\), and hence \(\Gamma_s(\eta) \le \Gamma_s(\eta_\nu)\).
Combining this with the previous estimate gives
\[
\Gamma_s(\eta) \le \Gamma_s(\eta_\nu) \le e^{\eta_\nu Q_\nu^{1/\sigma}} \le e^{\eta Q_\nu^{1/\sigma}} \qquad (\nu \ge 0).
\]
Choosing \(\nu = 0\) (so that \(Q_0 = T\)) yields the desired bound
\[
\Gamma_s(\eta) \le e^{\eta T^{1/\sigma}},
\]
with \(\eta = \wp + s c\).

\end{proof}

The central idea of our method rests on the precise interplay between Gevrey regularity and the non-resonance condition. The subexponential decay \(e^{-\eta Q^{1/\sigma}}\), inherent to the Gevrey class, precisely counterbalances the growth of the approximation function \(\Delta(Q)\) allowed by the \(\sigma\)-Bruno-R\"ussmann condition. This balance is quantified by the boundedness of the auxiliary function \(\Gamma_s(\eta)\), which thereby guarantees the control of small divisors in solving the cohomological equation~\eqref{h}.

\section{Quantization Procedure}\label{sec3}

The proof of Theorem \ref{ac} is structured in several steps.
First, we construct a Fourier integral operator \(T_h(t) = T_{1h}(t) \circ T_{2h}(t)\) that conjugates \(P_h(t)\) to a family of semiclassical pseudodifferential operators \(\widetilde{P}_h(t)\).
Subsequently, by means of an elliptic pseudodifferential operator \(A_h(t)\), we conjugate \(\widetilde{P}_h(t)\) to an operator whose principal symbol is \(K_0(I;t)+R_0(\varphi,I;t)\) and whose subprincipal symbol vanishes.
All operations on pseudodifferential operators are performed at the level of their symbols; the detailed symbolic calculus is presented in the proof of Theorem \ref{ad}.
Finally, we observe that the operator \(U_h(t)\) appearing in the statement of Theorem \ref{ac} is built from \(T_{1h}(t)\) and \(T_{2h}(t)\), reflecting the successive conjugations described above.

\subsection{Construction of the operator $T_{1h}$}

We construct the Fourier integral operator \(T_{1h}(t)\) as the quantization of the symplectomorphism \(\chi^1_t\), following the approaches in \cite{MR1770800} and \cite{MR0501196}.

Let \(\mathbb{L}\) be a smooth Hermitian line bundle over \(\mathbb{T}^n\) associated with a cohomology class \(\vartheta \in H^1(\mathbb{T}^n;\mathbb{Z}) \cong \mathbb{Z}^n\) via the representation
\[
\mathbb{Z} \longrightarrow \mathrm{SU}(1), \qquad n \longmapsto e^{{\rm i}\frac{\pi}{2}n}.
\]
The Maslov class of the Lagrangian submanifold \(\Lambda_\omega\) (for \(\omega \in \Omega_\sigma\)), defined as in \cite{MR211415}, can be canonically identified with an element \(\vartheta \in H^1(\mathbb{T}^n;\mathbb{Z})\) through the symplectic map \(\chi_t = \chi^0_t \circ \chi^1_t\).
Following \cite{MR0501196}, we consider the flat Hermitian line bundle \(\mathbb{L}\) over \(\mathbb{T}^n\) determined by \(\vartheta\). Its sections are conveniently described by functions on the universal cover \(\mathbb{R}^n\): a section \(f\) of \(\mathbb{L}\) corresponds uniquely to a function \(\tilde{f} : \mathbb{R}^n \to \mathbb{C}\) satisfying the quasi-periodicity condition
\begin{align} \label{ax}
\tilde{f}(x + 2\pi p) = e^{{\rm i}\frac{\pi}{2}\,\langle \vartheta,\,p \rangle} \,\tilde{f}(x),
\qquad \forall\,x\in\mathbb{R}^n,\; p\in\mathbb{Z}^n.
\end{align}
An orthonormal basis of \(L^2(\mathbb{T}^n;\mathbb{L})\) is then given by the family \(\{e_m\}_{m\in\mathbb{Z}^n}\), where
\[
\widetilde{e}_m(x) = e^{{\rm i}\langle m + \vartheta/4,\, x \rangle}.
\]

We now construct the Fourier integral operator $T_{1h}(t)$ associated with the Lagrangian manifold
\[
L_1 = \{(x,\xi,y,\eta): (x,\xi) = \chi^1_t(y,\eta),\ (y,\eta) \in \mathbb{T}^n \times D\},
\]
which maps $C^\infty(\mathbb{T}^n; \Omega^{1/2} \otimes \mathbb{L})$ to $C_0^\infty(M, \Omega^{1/2})$. Here $\Omega^{1/2}$ denotes the half-density bundle, and sections of $\mathbb{L}$ are defined by the quasi-periodicity condition \eqref{ax}. Since \(\chi_{t}^{1}: \mathbb{T}^{n}\times D \to T^{*}M\) (with \(D \subset \mathbb{R}^{n}\)) is an exact symplectomorphism, the twisted conormal
\[
L_{1}' = \{(x,y,\xi,\eta) : (x,\xi,y,-\eta) \in L_{1}\}
\]
is an exact Lagrangian submanifold of \(T^{*}(M \times \mathbb{T}^{n})\).
 This means that the pull-back \(\hat{\iota}^{*}\alpha\) of the canonical one-form \(\alpha\) on \(T^{*}(M \times \mathbb{T}^{n})\) under the inclusion \(\hat{\iota}: L_{1}' \hookrightarrow T^{*}(M \times \mathbb{T}^{n})\) satisfies
\[
\hat{\iota}^{*}\alpha = {\rm d}f
\]
for a smooth function \(f\) on \(L_{1}'\).

Set
\[
\phi(y_0, \xi_0; t_0) = \langle x_0, \xi_0 \rangle - f(\zeta; t_0), \quad \zeta = (x_0, y_0, \xi_0, \eta_0) \in L_1',
\]
where $f$ is given as above.
Consider the real analytic phase function
\[
\Psi(x,y,\xi;t) = \langle x, \xi \rangle - \phi(y,\xi;t),
\]
which locally parametrizes $L_1'$. Indeed, on the critical set
\[
O_{\Psi}=\{(x,y,\xi): {\rm d}_{\xi}\Psi=0\}
\]
we have $\operatorname{rank} {\rm d}_{(x,y,\xi)}{\rm d}_{\xi}\Psi = n$, and the map
\[
\hat{\iota}_{\Psi} \colon O_{\Psi} \longrightarrow L_{1}^{\prime}, \qquad
(x,y,\xi) \longmapsto \bigl(x,y,\Psi_{x},\Psi_{y}\bigr)
\]
is a local diffeomorphism onto an open subset of $L_{1}^{\prime}$.

Fix $\bar{\sigma} > 1$ and choose a symbol $a \in \mathcal{S}_l(U \times U_2)$, where $U = U_0 \times U_1$ with $U_0$ a neighborhood of $x_0$ and $(y,\xi) \in U_1 \times U_2$. We extend $a$ to $y \in \mathbb{R}^n$ by
\[
\tilde{a}(x, y + 2\pi p, \xi, h; t) = e^{{\rm i} \frac{\pi}{2} \langle \vartheta, p \rangle} a(x, y, \xi, h; t), \qquad (x,y,\xi) \in U \times \mathbb{R}^n,\ p \in \mathbb{Z}^n,
\]
and extend $\phi$ to a $2\pi$-periodic function in the variable  $y$ on $U_0 \times (U_1 + 2\pi \mathbb{Z}^n) \times U_2$. As in \cite [Proposition 1.3.1]{MR0405513}, one may consider more general phase functions. Given a section $u \in C^\infty(\mathbb{T}^n; \mathbb{L})$ of the line bundle \(\mathbb{L}\), define
\begin{align}\label{ay}
T_h(t) u(x) = (2\pi h)^{-n} \int_{\mathbb{R}^n} \int_{U_1} e^{{\rm i} \Psi(x,y,\xi;t)/h} \tilde{a}(x,y,\xi,h;t) \tilde{u}(y)  {\rm d} \xi  {\rm d}y,
\end{align}
where $\tilde{u}$ satisfies \eqref{ax}. Note that
\[
\tilde{a}(x, y+2\pi, \xi, h; t) \tilde{u}(y+2\pi) = \tilde{a}(x, y, \xi, h; t) \tilde{u}(y),
\]
so that the integrand \(\tilde{a}(x, y, \xi, h; t) \tilde{u}(y)\) is \(2\pi\)-periodic in \(y \in \mathbb{R}^{n}\). We denote the distribution kernel of \(T_{h}(t)\) by \(K_{h}(\Psi,a)\) and define a class of \(h\)-Fourier integral operators
\[
T_h(t): C^\infty(\mathbb{T}^n; \Omega^{1/2} \otimes \mathbb{L}) \to C_0^\infty(M, \Omega^{1/2}),
\]
where $T_h(t) = \sum_{i=1}^N T_h^i(t)$ for a finite integer \(N\), and each \(T_{h}^{i}(t)\) possesses a
\(\mathcal{G}^{\bar{\sigma}}\)-symbol.  Locally, each \(T_{h}^{i}(t)\) is given by the expression~\eqref{ay}.

We denote by $I^{\bar{\sigma}}(M \times \mathbb{T}^n, L_1'; \Omega^{1/2} \otimes \mathbb{L}', h)$ the class of distribution kernels $K_h$ of the operators $T_h$ (without $t$-dependence), where $\mathbb{L}'$ is the dual bundle to $\mathbb{L}$. Note that the definition depends on the choice of phase functions.
According to \cite{MR0405513}, the principal symbol of $T_h$ is of the form $e^{{\rm i} f(\zeta)/h} \Xi(\zeta)$, where $\Xi$ is a smooth section of
of the bundle
\[
\Omega^{\frac{1}{2}}(L_{1}') \otimes M_{L} \otimes \pi_{L}^{*}(\mathbb{L}'),
\]
in which $\Omega^{\frac{1}{2}}(L_{1}')$ is the half-density bundle on $L_{1}'$,
$M_{L}$ is the Maslov bundle of $L_{1}'$, and $\pi_{L}^{*}(\mathbb{L}')$ is the
pull-back of $\mathbb{L}'$ via the canonical projection
$\pi_{L} : L_{1}' \to \mathbb{T}^{n}$.  The bundle $\Omega^{\frac{1}{2}}(L_{1}')$ is trivialized by pulling back the
canonical half-density on $\mathbb{T}^{n} \times D$ under the projection
$\pi_{2} : L_{1}' \to \mathbb{T}^{n} \times D$.
As in the proof of Theorem~2.5 of \cite{MR0501196}, the bundle
$\pi_{L}^{*}(\mathbb{L}')$ can be canonically identified with the dual
$M_{L}'$ of the Maslov bundle.  Hence the principal symbol of $T_h$ may be
canonically identified with a smooth function $b$ on $L_{1}'$.
Moreover, for any operator $T_h$ of the form~\eqref{ay}, one has
\[
b\bigl(\phi_{\xi}'(y,\xi), y, \xi, -\phi_{y}'(y,\xi)\bigr)
= a_{0}\bigl(\phi_{\xi}'(y,\xi), y, \xi\bigr)\,
\bigl|\det\bigl(\partial^{2}\phi(y,\xi)/\partial y\,\partial\xi\bigr)\bigr|^{-1/2},
\]
where $a_{0}$ denotes the leading term of the amplitude $a$.

From the foregoing construction we select an operator \(T_{1h}(t)\) whose principal symbol equals \(1\) in a neighbourhood of the pull-back, via \(\pi_{2}\), of the set \(\Lambda\) associated with \(H\circ\chi^{1}_{t}\).
We write \(T_{1h}(t)=A_{h}(t)+h B_{h}(t)\).
Furthermore, we assume that the leading term \(a_{0}\) equals \(1\) on \(U\); this is the principal symbol of the elliptic operator \(A_{h}(t)\).
We then solve a linear equation for the real part of the principal symbol of \(B_{h}(t)\).
This hypothesis forms the foundation for the subsequent analysis.
We conjugate the original operator \(P_{h}(t)\) by \(T_{1h}(t)\) (defined as above).
Applying the stationary phase Lemma~\ref{az}, we prove that
\[
P_{h}^{1}(t)=T_{1h}^{*}(t)\circ P_{h}(t)\circ T_{1h}(t)
\]
is an \(h\)-pseudodifferential operator on \(C^{\infty}(\mathbb{T}^{n};\mathbb{L})\) with a symbol belonging to \(\mathcal{S}^{\bar{\sigma}}(\mathbb{T}^{n}\times D)\).
Lemma~2.9 of \cite{MR0501196} then implies that the principal symbol of \(P_{h}^{1}(t)\) is \(H\circ\chi^{1}_{t}\) and that its subprincipal symbol vanishes.

We now state the stationary phase lemma.
\begin{lemma}\label{az}
Let \(\Phi(x,y)\) be a real analytic function in a neighborhood of \((0,0)\) in \(\mathbb{R}^{m_1+m_2}\).
Assume that \(\partial_x\Phi(0,0)=0\) and that the Hessian \(\partial_{xx}^2\Phi(0,0)\) is non-singular.
Let \(x(y)\) be the unique solution of \(\partial_x\Phi(x,y)=0\) with \(x(0)=0\), given by the implicit function theorem.
Then, for any symbol \(g \in \mathcal{S}^{\bar{\sigma}}(U)\), where \(U\) is a suitable neighborhood of \((0,0)\), we have
\[
\int e^{{\rm i} \Phi(x,y)/h} g(x,y,h)  {\rm d} x = e^{{\rm i} \Phi(x(y),y)/h} G(y,h),
\]
where $G \in \mathcal{S}^{\bar{\sigma}} (V)$ for some neighborhood \(V\) of \(0\) in \(\mathbb{R}^{m_2}\).
\end{lemma}
The proof of this lemma can be found in \cite{MR1770800}.

\subsection{Construction of the operator $T_{2h}$}

This section is devoted to the construction of the Fourier integral operator \(T_{2h}(t)\).
Our starting point is the generating function \(\Phi\) of the symplectomorphism \(\chi^{0}_{t}\).
The procedure for obtaining this function is described in detail  in \cite{MR4578527}.
We assume that the generating function \(\Phi(x,I;t)\) belongs to the Gevrey class
\(\mathcal{G}^{\rho,\rho+1,\rho+1}\bigl(\mathbb{R}^{n}\times D\times(-\tfrac12,\tfrac12)\bigr)\) and has the form
\begin{align} \label{a9}
\Phi(x,I;t)=\langle x,I\rangle + \phi(x,I;t),
\end{align}
where \(\phi\) is \(2\pi\)-periodic in \(x\) and satisfies
\(\phi\in\mathcal{G}^{\rho,\rho+1,\rho+1}\bigl(\mathbb{T}^{n}\times D\times(-\tfrac12,\tfrac12)\bigr)\).
Moreover, $\phi$ satisfies the estimates
\begin{align}\label{bs}
\left| D_I^\alpha D_\theta^\beta D_t^\delta \phi(\theta,I;t) \right| \leq A \kappa C^{|\alpha|+|\beta|+|\delta|} L_1^{|\alpha|} (L_1/\kappa)^{|\beta|} \alpha!^{\rho+1} \beta!^\rho \delta!^{\rho+1} L_1^{N/2} \varepsilon_H^q
\end{align}
for $(\theta,I;t) \in \mathbb{T}^n \times D \times (-\tfrac{1}{2}, \tfrac{1}{2})$ and with constants \(A,\kappa,C,L_{1}>0\).   Here $\varepsilon_H^q$ (with $0 < q < 1$) denotes the norm of the perturbation $H$; we typically take $q = \tfrac{1}{2}$.

For \(\varepsilon_{H}\) sufficiently small, the function \(\Phi\) is a generating function of an exact symplectic transformation
\[
\chi^0_t: \mathbb{T}^n \times D \times (-\tfrac{1}{2}, \tfrac{1}{2}) \to \mathbb{T}^n \times D,
\]
which itself belongs to the Gevrey class \(\mathcal{G}^{\rho,\rho+1,\rho+1}\).   The transformation is defined implicitly by
\[
\chi^0_t(\nabla_I \Phi(\theta,I;t), I) = (\theta, \nabla_\theta \Phi(\theta,I;t)).
\]
Under this transformation the Hamiltonian becomes
\begin{align}\label{a5}
\tilde{H}(\varphi,I;t) = (H \circ \chi_t^0)(\varphi,I) = K_0(I;t) + R_0(\varphi,I;t), \quad \omega \in \Omega_\sigma.
\end{align}

We now construct the operator \(T_{2h}(t)\) by quantizing the symplectomorphism \(\chi^{0}_{t}\).
The resulting \(h\)-Fourier integral operator
\[
T_{2h}(t): L^{2}(\mathbb{T}^{n};\mathbb{L}) \longrightarrow L^{2}(\mathbb{T}^{n};\mathbb{L})
\]
is associated with the canonical relation graph of \(\chi^{0}_{t}\).
Its distribution kernel has the form
\begin{align} \label{a6}
(2\pi h)^{-n} \int_{D} e^{i (\langle x-y,I\rangle + \phi(x,I;t))/h}
b(x,I;h,t) \; \mathrm{d}I, \qquad
x \in \mathbb{T}^{n},\; I \in D,
\end{align}
where the amplitude \(b\) belongs to the Gevrey symbol class
\(\mathcal{S}_{\tilde{l}}\bigl(\mathbb{T}^{n}\times D\times(-\frac12,\frac12)\bigr)\)
with \(\tilde{l}=(\sigma,\mu,\lambda,\sigma+\mu+\lambda-1)\).
Here \(\phi\) is the same function as in~\eqref{a9}.
We assume that the principal symbol of \(T_{2h}(t)\) equals \(1\) in a neighbourhood of
\(\mathbb{T}^{n}\times D\).

Define the composite operator $T_h(t) = T_{1h} \circ T_{2h}(t)$. Then
\[
\widetilde{P}_{h}(t)=T_{h}^{*}(t) \circ P_{h}(t) \circ T_{h}(t)
= T_{2h}^{*}(t) \circ P_{h}^{1}(t) \circ T_{2h}(t),
\]
where \(P_{h}^{1}(t)=T_{1h}^{*}(t) \circ P_{h}(t) \circ T_{1h}(t)\).
Applying Definition~\ref{bv} together with the stationary phase Lemma~\ref{az},
one readily verifies that \(\widetilde{P}_{h}(t)\) is an \(h\)-pseudodifferential operator
whose symbol lies in the class \(\mathcal{S}_{\tilde{l}}\),
again with \(\tilde{l}=(\sigma,\mu,\lambda,\sigma+\mu+\lambda-1)\).

\subsection{Construction of $U_h(t)$}

Suppose that \(A_{h}(t)\) is an elliptic pseudodifferential operator with symbol \(a \in \mathcal{S}_{l}\) and principal symbol \(a_{0}=1\). Define the operator $V_h(t) = T_h(t) \circ A_h(t)$, where $T_h(t) = T_{1h}(t) \circ T_{2h}(t)$. This construction yields the relation
\[
P_h(t) \circ V_h(t) = V_h(t) \circ (P_h^0(t) + R_h(t)),
\]
where $P_h^0(t)$ and $R_h(t)$ possess the desired analytical properties. However, \(V_{h}(t)\) need not be unitary.

To obtain a unitary operator, we define
\[
Q_h(t) = (V_h^*(t) \circ V_h(t))^{-1/2},
\]
which is a pseudodifferential operator with symbol $q(\varphi, I; h, t)$ in the class $\mathcal{S}_{\tilde{l}}(\mathbb{T}^n \times D \times (-\tfrac{1}{2}, \tfrac{1}{2}))$, where $\tilde{l} = (\sigma, \mu, \lambda, \sigma + \mu + \lambda - 1)$. The desired unitary operator is then given by
\[
U_h(t) = V_h(t) \circ Q_h(t).
\]
This operator $U_h(t)$ satisfies all required properties for our analysis.

\section{ Homological Equation under the $\sigma$-Bruno-R\"ussmann Condition in Gevrey Classes}
\label{sec4}

To construct the quantum Birkhoff normal form, we solve the homological equation on the Cantor set
\(\mathbb{T}^{n}\times E_{\kappa}(t)\), where \(E_{\kappa}(t)\subset D\) consists of action variables
satisfying the \(\sigma\)-Bruno-R\"ussmann condition.
The frequency map \(I\longmapsto\omega(I;t)\) inherits Gevrey regularity from the Hamiltonian \(H\).

The notation $\mathcal{S}_{l}\bigl(\mathbb{T}^{n}\times E_{\kappa}(t)\times(-\tfrac12,\tfrac12)\bigr)$
requires careful explanation.
The parameter $t\in(-\tfrac12,\tfrac12)$ plays a dual role:
it determines both the time-dependent Hamiltonian and the associated non-resonant set $E_{\kappa}(t)$,
while simultaneously appearing as an independent variable in the Gevrey symbol classes.
This double interpretation is essential for following the $t$-dependence of the conjugating operators
and for establishing \emph{uniform} Gevrey estimates.
In all derivative bounds, $t$ and $I$ are treated as independent variables; the symbol spaces on the
parameter-dependent Cantor set $\mathbb{T}^{n}\times E_{\kappa}(t)$ are understood via Whitney
extension theory.

We now assume that \(\widetilde{P}_{h}(t)\) is a self-adjoint \(h\)-pseudodifferential operator with symbol
\(p \in \mathcal{S}_{\tilde{l}}\bigl(\mathbb{T}^{n}\times D\times(-\frac12,\frac12)\bigr)\),
where \(\tilde{l}=(\sigma,\mu,\lambda,\sigma+\mu+\lambda-1)\).  Its asymptotic expansion is
\[
p(\varphi,I;t,h) \sim \sum_{j=0}^{\infty} p_{j}(\varphi,I;t) h^{j},
\]
and on the Cantor set we require
\[
p_{0}(I;t)=K_{0}(I;t), \qquad p_{1}(I;t)=0, \qquad
\forall\; (\varphi,I;t)\in \mathbb{T}^{n}\times E_{\kappa}(t)\times(-\tfrac12,\tfrac12).
\]

We shall conjugate \(\widetilde{P}_{h}(t)\) to a normal form \(P_{h}^{0}(t)\) by an elliptic
pseudodifferential operator \(A_{h}(t)\) whose symbol
\(a(\varphi,I;t,h)\) belongs to \(\mathcal{S}_{l}\bigl(\mathbb{T}^{n}\times E_{\kappa}(t)\bigr)\)
with \(l=(\sigma,\mu,\lambda,\bar{\rho})\).  Explicitly,
\[
P_{h}^{0}(t)=A_{h}^{*}(t)\circ \widetilde{P}_{h}(t)\circ A_{h}(t).
\]
By construction, the symbol \(p\) of \(\widetilde{P}_{h}(t)\) lies in the class
\(\mathcal{S}_{\tilde{l}}\bigl(\mathbb{T}^{n}\times E_{\kappa}(t)\bigr)\) with the same index
\(\tilde{l}=(\sigma,\mu,\lambda,\sigma+\mu+\lambda-1)\).

Theorem~\ref{ac} follows directly from Theorem~\ref{ad}; it therefore suffices to prove
Theorem~\ref{ad}.  In other words, we must solve accurately the homological equation
\(\mathcal{L}_{\omega}u=f\) under the \(\sigma\)-Bruno-R\"ussmann condition within the
Gevrey classes on the parameter-dependent Cantor set \(E_{\kappa}(t)\), as described in the
Introduction.

\begin{theorem}\label{ad}
There exist symbols
\[
a,\ p^{0} \in \mathcal{S}_{l}\bigl(\mathbb{T}^{n}\times E_{\kappa}(t)\times(-\tfrac12,\tfrac12)\bigr),
\qquad l=(\sigma,\mu,\lambda,\bar{\rho}),
\]
with asymptotic expansions
\[
a(\varphi,I;t,h) \sim \sum_{j=0}^{\infty} a_{j}(\varphi,I;t)h^{j}, \qquad
p^{0}(I;t,h) \sim \sum_{j=0}^{\infty} p_{j}^{0}(I;t)h^{j},
\]
satisfying \(a_{0}=1\), \(p_{0}^{0}=K_{0}(I;t)\) and \(p_{1}^{0}=0\), such that
\[
p \circ a \;-\; a \circ p^{0} \;\sim\; 0
\]
in \(\mathcal{S}_{l}\bigl(\mathbb{T}^{n}\times E_{\kappa}(t)\times(-\tfrac12,\tfrac12)\bigr)\).
Here \(\circ\) denotes the composition of symbols (i.e. the star product) and \(\sim\) means equality of formal power series in \(h\).
\end{theorem}

\subsection{Solving the Homological Equation}

This section is devoted to solving the homological equation \(\mathcal{L}_{\omega}u = f\).
We first fix some notation. For a multi-index \(\gamma = (\gamma_1, \dots, \gamma_n) \in \mathbb{Z}_+^n\) and a vector \(k = (k_1, \dots, k_n) \in \mathbb{Z}_+^n\), we set
\begin{equation*}\label{a3}
\kappa^{\gamma} = \prod_{i=1}^n \kappa_i^{\gamma_i}, \qquad |k| = |k_1| + \dots + |k_n|.
\end{equation*}
The partial order on \(\mathbb{Z}_+^n\) is defined componentwise:
\begin{equation*}
\alpha \leq \gamma \quad \Longleftrightarrow \quad \alpha_i \leq \gamma_i \ \text{ for every } 1 \leq i \leq n.
\end{equation*}

The following lemma provides the key estimates needed for the solution.

\begin{lemma}
Let \(\omega(I;t) \in C^{\infty}\bigl(E_{\kappa}(t) \times (-\tfrac12,\tfrac12); \mathbb{R}^{n}\bigr)\) satisfy the following Gevrey-type estimates:
\begin{align}
\bigl| D^{\alpha}_{I} D^{\delta}_{t} \omega(I;t) \bigr| &\le C_{1}^{|\alpha|+|\delta|}\, \alpha!^{\mu+1}\, \delta!^{\mu+1},
&& \forall\, I \in E_{\kappa}(t),\; \alpha \in \mathbb{Z}_{+}^{n},\; \delta \in \mathbb{Z}_{+}, \label{bc} \\
\bigl| \langle \omega(I;t), k \rangle \bigr| &\ge \frac{\kappa}{\Delta(|k|)},
&& \forall\, I \in E_{\kappa}(t),\; k \in \mathbb{Z}^{n} \setminus \{0\}. \label{bd}
\end{align}
Then there exists a constant \(C_{0}>0\), depending only on \(n\), \(\kappa\), and \(C_{1}\), such that for every
\(I \in E_{\kappa}(t)\), every non-zero \(k \in \mathbb{Z}^{n}\), and every multi-index \(\alpha \in \mathbb{Z}_{+}^{n}\) and every \(\delta \in \mathbb{Z}_{+}\),
\begin{align}
\bigl| D^{\alpha}_{I} D^{\delta}_{t} \bigl( \langle \omega(I;t), k \rangle^{-1} \bigr) \bigr|
&\le C_{0}^{|\alpha|+|\delta|+1} \,
\alpha! \, \delta! \,
\max_{0 \le j \le |\alpha|+|\delta|}
\Bigl( (|\alpha|+|\delta|-j)!^{\mu} \,
|k|^{j} \, \Delta^{\,j+1}(|k|) \Bigr). \label{ba}
\end{align}
\end{lemma}

\begin{proof}
Denote \(l_{k}(I;t) = \langle \omega(I;t), k \rangle\) for \(0 \neq k \in \mathbb{Z}^{n}\).
Applying Leibniz' rule to the identity \(l_{k}(I;t) \cdot l^{-1}_{k}(I;t) = 1\), we obtain for \(|\alpha| \ge 1\) or \(|\delta| \ge 1\)
\begin{align*}
0 = D^{\alpha}_{I} D^{\delta}_{t} \bigl( l_{k}(I;t) \, l^{-1}_{k}(I;t) \bigr)
= \sum_{\alpha' \le \alpha} \sum_{\delta' \le \delta}
\binom{\alpha}{\alpha'} \binom{\delta}{\delta'}
\bigl( D^{\alpha'}_{I} D^{\delta'}_{t} l_{k}(I;t) \bigr)
\bigl( D^{\alpha-\alpha'}_{I} D^{\delta-\delta'}_{t} l^{-1}_{k}(I;t) \bigr).
\end{align*}
Separating the term with \((\alpha',\delta')=(0,0)\) gives
\begin{align*}
- D^{\alpha}_{I} D^{\delta}_{t} l^{-1}_{k}(I;t)
= l^{-1}_{k}(I;t) \Bigg[
& \sum_{\substack{0 \le \alpha' \le \alpha \\ 0 \le \delta' \le \delta \\ (\alpha',\delta') \neq (0,0)} }
\binom{\alpha}{\alpha'} \binom{\delta}{\delta'}
\bigl( D^{\alpha'}_{I} D^{\delta'}_{t} l_{k}(I;t) \bigr)
\bigl( D^{\alpha-\alpha'}_{I} D^{\delta-\delta'}_{t} l^{-1}_{k}(I;t) \bigr) \Bigg].
\end{align*}
We decompose the sum into three parts:
\begin{align*}
I_{1} &= \sum_{\substack{0 < \alpha' \le \alpha \\ 0 < \delta' \le \delta}}
\binom{\alpha}{\alpha'} \binom{\delta}{\delta'}
\bigl( D^{\alpha'}_{I} D^{\delta'}_{t} l_{k}(I;t) \bigr)
\bigl( D^{\alpha-\alpha'}_{I} D^{\delta-\delta'}_{t} l^{-1}_{k}(I;t) \bigr), \\
I_{2} &= \sum_{0 < \delta' \le \delta} \binom{\delta}{\delta'}
\bigl( D^{\delta'}_{t} l_{k}(I;t) \bigr)
\bigl( D^{\alpha}_{I} D^{\delta-\delta'}_{t} l^{-1}_{k}(I;t) \bigr), \\
I_{3} &= \sum_{0 < \alpha' \le \alpha} \binom{\alpha}{\alpha'}
\bigl( D^{\alpha'}_{I} l_{k}(I;t) \bigr)
\bigl( D^{\alpha-\alpha'}_{I} D^{\delta}_{t} l^{-1}_{k}(I;t) \bigr).
\end{align*}
Thus
\[
- D^{\alpha}_{I} D^{\delta}_{t} l^{-1}_{k}(I;t)
= l^{-1}_{k}(I;t) \bigl( I_{1} + I_{2} + I_{3} \bigr).
\]

Assume that estimate \eqref{ba} holds for all multi-indices with \(|\alpha|+|\delta| < p\).
We shall prove it for \(|\alpha|+|\delta| = p\).

From \eqref{bc} we obtain a constant \(C_{2} > 0\) (depending only on \(C_{1}\)) such that for all
\(I \in E_{\kappa}(t)\) and all non-zero \(\alpha,\delta\),
\begin{equation} \label{bc_prime}
\bigl| D^{\alpha}_{I} D^{\delta}_{t} \omega(I;t) \bigr|
\le C_{2}^{|\alpha|+|\delta|}
\Bigl( \frac{\alpha!}{|\alpha|} \Bigr)^{\mu+1}
\Bigl( \frac{\delta!}{|\delta|} \Bigr)^{\mu+1}.
\end{equation}
 Set \(\varepsilon = C_{0}^{-1} C_{2}\) (the constant \(C_{0}\) will be chosen later).
Using \eqref{bd}, the induction hypothesis, and the elementary inequality
\(x! \, y! \le (x+y)!\) for non-negative integers, we estimate each term.

\noindent \textbf{Estimate of \(I_{1}\):}
\begin{align*}
|I_{1}| &\le \frac{\Delta(|k|)}{\kappa} \,
\sum_{\substack{0 < \alpha' \le \alpha \\ 0 < \delta' \le \delta}}
\binom{\alpha}{\alpha'} \binom{\delta}{\delta'}
\Bigl| D^{\alpha'}_{I} D^{\delta'}_{t} l_{k}(I;t) \Bigr|
\Bigl| D^{\alpha-\alpha'}_{I} D^{\delta-\delta'}_{t} l^{-1}_{k}(I;t) \Bigr| \\
&\le \frac{\Delta(|k|)}{\kappa} \,
\sum_{\substack{0 < \alpha' \le \alpha \\ 0 < \delta' \le \delta}}
\binom{\alpha}{\alpha'} \binom{\delta}{\delta'}
C_{2}^{|\alpha'|+|\delta'|}
\Bigl( \frac{\alpha'!}{|\alpha'|} \Bigr)^{\mu}
\Bigl( \frac{\delta'!}{|\delta'|} \Bigr)^{\mu}
|k| \\
&\qquad \times C_{0}^{|\alpha-\alpha'|+|\delta-\delta'|+1}
(\alpha-\alpha')! \, (\delta-\delta')! \\
&\qquad \times \max_{0 \le j \le |\alpha-\alpha'|+|\delta-\delta'|}
\Bigl( (|\alpha-\alpha'|+|\delta-\delta'|-j)!^{\mu}
|k|^{j} \Delta^{\,j+1}(|k|) \Bigr).
\end{align*}
Using \(\binom{\alpha}{\alpha'} \alpha'! \, (\alpha-\alpha')! = \alpha!\) (and similarly for \(\delta\)),
and noting that the maximum in the last line is bounded by
\[
\max_{0 \le j \le |\alpha|+|\delta|-1}
\Bigl( (|\alpha|+|\delta|-j-1)!^{\mu} |k|^{j+1} \Delta^{\,j+2}(|k|) \Bigr),
\]
we obtain after simplification
\begin{align*}
|I_{1}| &\le
C_{0}^{|\alpha|+|\delta|+1} \alpha! \, \delta! \,
\max_{0 \le j \le |\alpha|+|\delta|}
\Bigl( (|\alpha|+|\delta|-j)!^{\mu} |k|^{j} \Delta^{\,j+1}(|k|) \Bigr) \\
&\qquad \times \Bigl( \kappa^{-1} \sum_{0 < \alpha'} \varepsilon^{|\alpha'|}
\sum_{0 < \delta'} \varepsilon^{|\delta'|} \Bigr).
\end{align*}
The double sum can be made smaller than \(1\) by choosing \(\varepsilon\) sufficiently small.
Indeed, there exists a constant \(c_{\varepsilon}\) (depending only on \(n\) and \(\varepsilon\)) such that
\[
\kappa^{-1} \sum_{0 < \alpha'} \varepsilon^{|\alpha'|}
\sum_{0 < \delta'} \varepsilon^{|\delta'|}
\le \varepsilon \kappa^{-1}
\Bigl( \sum_{s=2}^{\infty} s^{n} \varepsilon^{s-2} \Bigr)
\Bigl( \sum_{\delta'=1}^{\infty} \varepsilon^{\delta'} \Bigr)
< 1
\]
for \(\varepsilon\) small enough.

\noindent \textbf{Estimates of \(I_{2}\) and \(I_{3}\):}
Analogous calculations yield the same upper bound for \(|I_{2}|\) and \(|I_{3}|\).

Collecting the three estimates and using \(|l_{k}^{-1}(I;t)| \le \kappa^{-1} \Delta(|k|)\), we finally obtain
\begin{align*}
\bigl| D^{\alpha}_{I} D^{\delta}_{t} l^{-1}_{k}(I;t) \bigr|
&\le 3 C_{0}^{|\alpha|+|\delta|+1} \alpha! \, \delta! \,
\max_{0 \le j \le |\alpha|+|\delta|}
\Bigl( (|\alpha|+|\delta|-j)!^{\mu} |k|^{j} \Delta^{\,j+1}(|k|) \Bigr).
\end{align*}
This completes the induction and proves the lemma.

\end{proof}

For any $l>0$, we set $\langle k\rangle_{l}=1+|k_{1}|^{l}+\cdots+|k_{n}|^{l},k\in \mathbb{Z}^{n}$. This function has the following properties:
\begin{align*}
|k|^{l}\leq n^{l}\langle k\rangle_{l}.
\end{align*}

Suppose that \(f(\varphi, I; t) \in C^{\infty}\bigl(\mathbb{T}^{n} \times E_{\kappa}(t) \times (-\frac{1}{2}, \frac{1}{2})\bigr)\)
satisfies the anisotropic Gevrey estimate
\begin{align} \label{f}
\bigl| D_{I}^{\alpha} D_{\varphi}^{\beta} D_{t}^{\delta} f(\varphi, I; t) \bigr|
\leq d_{0} C^{\mu |\alpha| + |\beta| + |\delta|}
\Gamma\bigl(\mu |\alpha| + \sigma |\beta| + \lambda |\delta| + q\bigr)
\end{align}
for all \(I \in E_{\kappa}(t)\), multi-indices \(\alpha, \beta \in \mathbb{Z}_{+}^{n}\),
\(\delta \in \mathbb{Z}_{+}\), and some \(q \geq 1\).
Here \(\Gamma\) denotes the Gamma function, and \(\sigma, \mu, \lambda > 0\) are fixed constants.
Assume further that \(f\) is \(2\pi\)-periodic in \(\varphi\).
Define its Fourier coefficients by
\begin{align} \label{bx}
f_{k}(I; t) = (2\pi)^{-n} \int_{\mathbb{T}^{n}}
e^{-i \langle k, \varphi \rangle} f(\varphi, I; t) \, \mathrm{d}\varphi, \qquad k \in \mathbb{Z}^{n}.
\end{align}
Then, for \(\sigma > 1\) and a constant \(C^{-1} > 0\), we have
\begin{align}\label{c}
\bigl| D_{I}^{\alpha} D_{t}^{\delta} f_{k}(I; t) \bigr|
\leq d_{0} C^{\mu |\alpha| + |\delta| + 1}
\Gamma\bigl(\mu |\alpha| + \lambda |\delta| + q\bigr)
e^{-C^{-1} |k|^{1/\sigma}},
\end{align}
for some \(q \geq 1\).  The proof of this estimate is given in Appendix \ref{A1}.

Then, integrating by parts, we also have the following estimate for $f_{k}$
\begin{align*}
|k^{\gamma}\langle k\rangle_{l}D_{I}^{\alpha}f_{k}|\leq d_{0}C^{\mu|\alpha|+|\gamma|+|\delta|+l}\Gamma(\mu|\alpha|+\sigma|\gamma|+\lambda|\delta|+\sigma l+q)e^{-c|k|^{1/\sigma}}
\end{align*}
for $\gamma\in \mathbb{Z}_{+},l\in \mathbb{N}$.

Now suppose
\begin{align}
\int_{\mathbb{T}^{n}} f(\varphi, I; t)  \mathrm{d} \varphi = 0. \label{g}
\end{align}
We consider the homological equation
\begin{align}\label{h}
 \mathcal{L}_{\omega} u(\varphi, I; t) = f(\varphi, I; t), \qquad u(0, I; t) = 0,
\end{align}
and we shall establish an upper bound for
\[
\bigl| D_{I}^{\alpha} D_{\varphi}^{\beta} D_{t}^{\delta} u(\varphi, I; t) \bigr|
\]
for all multi-indices \(\alpha, \beta \in \mathbb{Z}_{+}^{n}\) and every \(\delta \in \mathbb{Z}_{+}\).

\begin{proposition} \label{av}
Let $f \in C^{\infty}\bigl(\mathbb{T}^{n} \times E_{\kappa}(t) \times (-\tfrac12,\tfrac12)\bigr)$ satisfy
\eqref{f} and \eqref{g}, and assume $\lambda \geq \mu$.
Then equation \eqref{h} admits  a unique solution
$u \in  C^{\infty}\bigl(\mathbb{T}^{n} \times E_{\kappa}(t) \times (-\tfrac12,\tfrac12)\bigr)$,
and this solution satisfies the estimate
\begin{align}\label{xx}
\bigl| D_{I}^{\alpha} D_{\varphi}^{\gamma} D_{t}^{\delta} u(\varphi,I;t) \bigr|
\leq d_{0} \mathcal D \,
C^{\mu|\alpha|+|\gamma|+|\delta|}\,
\Gamma\bigl(\mu|\alpha|+\sigma|\gamma|+\lambda|\delta|+\sigma l+q\bigr)
\end{align}
for all $I \in E_{\kappa}(t)$, all multi-indices $\alpha,\gamma \in \mathbb{Z}_{+}^{n}$, and every
$\delta \in \mathbb{Z}_{+}$.
Here $\mathcal D$ is a constant depending on $n$, $C$, $C_{0}$, $\delta$, $\eta$, $s$, $T$, and $\mu$.

\end{proposition}

\begin{proof}
We begin by expressing \( f \) and \( u \) via their Fourier series. Since \( f \) is \( 2\pi \)-periodic in \( \varphi \) and satisfies \( \int_{\mathbb{T}^n} f(\varphi, I; t)  \mathrm{d} \varphi = 0 \) by \eqref{g}, we have
\[
f(\varphi, I; t) = \sum_{0 \neq k \in \mathbb{Z}^{n}} e^{\mathrm{i} \langle k, \varphi \rangle} f_k(I; t),
\]
where
\[
f_k(I; t) = (2\pi)^{-n} \int_{\mathbb{T}^{n}} e^{-\mathrm{i} \langle k, \varphi \rangle} f(\varphi, I; t)  \mathrm{d} \varphi.
\]
Similarly, we expand \( u \). Condition \eqref{g} implies \( u_0 = 0 \). Substituting these expansions into equation \eqref{h} and comparing coefficients yields, for every non-zero \( k \in \mathbb{Z}^n \) and \( I \in E_{\kappa}(t) \),
\[
u_k(I; t) = \langle \omega(I; t), k\rangle^{-1} f_k(I; t).
\]
Thus, a solution \( u \) exists and is given by
\[
u(\varphi, I; t) = \sum_{0 \neq k \in \mathbb{Z}^n} \langle \omega(I; t), k\rangle^{-1} f_k(I; t) e^{i\langle k, \varphi\rangle}.
\]
To establish the desired regularity, we estimate the derivatives of \( u \). Define $W(|k|)=|k|^{l}\Delta^{l+1}(|k|)$ for any real constant $l$.

 Applying the Leibniz rule and estimates \eqref{ba} and \eqref{c}, we obtain for any multi-indices \( \alpha, \gamma \in \mathbb{Z}_+^n \) with \( \alpha_1 \leq \gamma \) and \( \delta \in \mathbb{Z}_+ \) the estimate
 \begin{align*}
\bigl|D_{I}^{\alpha}D_{\varphi}^{\gamma}D_{t}^{\delta}u(\varphi,I;t)\bigr|
&
\\
\leq &\sum_{0 \neq k \in \mathbb{Z}^n} \sum_{\substack{0 < \alpha_1 \leq \alpha \\ 0 < \delta_1 \leq \delta}}
\binom{\alpha}{\alpha_1} \binom{\delta}{\delta_1}
\left| D_I^{\alpha_1} D_t^{\delta_1} \left( \langle \omega, k \rangle^{-1} \right) \right|
\cdot \left|k^{\gamma} D_I^{\alpha - \alpha_1} D_t^{\delta - \delta_1} f_k(I; t) \right|
\end{align*}
 Taking into account \eqref{ba} and \eqref{c}, we obtain for any multi-indices
$\alpha,\alpha_{1},\gamma\in\mathbb{Z}_{+}^{n}$
and any nonnegative integer $\delta$ the estimate
\[
\begin{aligned}
\bigl|D_{I}^{\alpha}D_{\varphi}^{\gamma}D_{t}^{\delta}u(\varphi,I;t)\bigr|
&\leq\Bigl|\sum_{0\neq k\in\mathbb{Z}^{n}}\sum_{0<\alpha_{1}\leq\alpha}
   \sum_{0<\delta_{1}\leq\delta}
   \binom{\alpha}{\alpha_{1}}\binom{\delta}{\delta_{1}}
   D_{I}^{\alpha_{1}}D_{t}^{\delta_{1}}\bigl(\langle\omega(I;t),k\rangle^{-1}\bigr)\\
&\qquad\times D_{I}^{\alpha-\alpha_{1}}D_{t}^{\delta-\delta_{1}}f_{k}(I;t)
   \,k^{\gamma}\Bigr|.
\end{aligned}
\]
Using the definition of $W(|k|)$ and the bounds provided by \eqref{ba} and
\eqref{c}, this yields
\[
\begin{aligned}
\bigl|D_{I}^{\alpha}D_{\varphi}^{\gamma}D_{t}^{\delta}u(\varphi,I;t)\bigr|
&\leq\sum_{0\neq k\in\mathbb{Z}^{n}}W(|k|)W(|k|)^{-1}
   \sum_{0<\alpha_{1}\leq\alpha}\sum_{0<\delta_{1}\leq\delta}
   \binom{\alpha}{\alpha_{1}}\binom{\delta}{\delta_{1}}
   C_{0}^{|\alpha_{1}|+|\delta_{1}|+1}\,\alpha_{1}!\,\delta_{1}!\\
&\quad\times\max_{0\leq j\leq|\alpha_{1}|+|\delta_{1}|}
   \bigl|(|\alpha_{1}|+|\delta_{1}|-j)!^{\mu}
   |k|^{j}\Delta^{j+1}(|k|)\bigr|\\
&\quad\times\bigl|k^{\gamma}
   D_{I}^{\alpha-\alpha_{1}}D_{t}^{\delta-\delta_{1}}f_{k}(I;t)\bigr|.
\end{aligned}
\]
By \eqref{c} we have
\[
|k^{\gamma}\langle k\rangle_{l}D_{I}^{\alpha}f_{k}|\leq d_{0}C^{\mu|\alpha|+|\gamma|+|\delta|+l}\Gamma(\mu|\alpha|+\sigma|\gamma|
+\lambda|\delta|+\sigma l+q)e^{-c|k|^{1/\sigma}}
\]
 hence
\[
\begin{aligned}
\bigl|D_{I}^{\alpha}D_{\varphi}^{\gamma}D_{t}^{\delta}u(\varphi,I;t)\bigr|
&\leq\sum_{0\neq k\in\mathbb{Z}^{n}}\Delta^{l+1}(|k|)
   e^{-C^{-1}|k|^{1/\sigma}}
   \sum_{0<\alpha_{1}\leq\alpha}\sum_{0<\delta_{1}\leq\delta}
   d_{0}C^{\mu|\alpha-\alpha_{1}|+|\gamma|+|\delta-\delta_{1}|+l+1}
   C_{0}^{|\alpha_{1}|+|\delta_{1}|+1}\\
&\quad\times\frac{\alpha!}{(\alpha-\alpha_{1})!}
   \frac{\delta!}{(\delta-\delta_{1})!}\,
   \Gamma\bigl(\mu|\alpha-\alpha_{1}|+\sigma|\gamma|+\lambda|\delta-\delta_{1}|
   +\sigma l+q\bigr)\\
&\quad\times\max_{0\leq j\leq|\alpha_{1}|+|\delta_{1}|}
   \bigl|(|\alpha_{1}|+|\delta_{1}|-j)!^{\mu}
   |k|^{j-l}\Delta^{j-l}(|k|)\bigr|.
\end{aligned}
\]

We now simplify the expression. The maximum in the last line is attained at $j=0$, yielding
\[
(|\alpha_1|+|\delta_1|)!^\mu\,|k|^{-l}\Delta^{-l}(|k|).
\]
By Stirling's formula, there exists a constant $C>0$ such that for every $x\ge1$,
\[
(x)!^\mu\le C_{3}^{\,x}\,\Gamma(\mu x).
\]
Using Lemma~\ref{bh} and the inequality $\Gamma(s)\Gamma(u)\le\Gamma(s+u)$ for $s,u\ge1$, we obtain
\[
\begin{aligned}
\bigl|D_{I}^{\alpha}D_{\varphi}^{\gamma}D_{t}^{\delta}u(\varphi,I;t)\bigr|
&\le d_{0}C_{0}C\sum_{0\neq k\in\mathbb{Z}^{n}}
   |k|^{-l}\Delta(|k|)e^{-C^{-1}|k|^{1/\sigma}}\\
&\quad\times\sum_{0<\alpha_{1}\le\alpha}\sum_{0<\delta_{1}\le\delta}
   \frac{\alpha!}{(\alpha-\alpha_{1})!}\frac{\delta!}{(\delta-\delta_{1})!}
   C^{\mu|\alpha-\alpha_{1}|+|\gamma|+|\delta-\delta_{1}|+l}
   C_{0}^{|\alpha_{1}|+|\delta_{1}|}\\
&\quad\times C_{3}^{|\alpha_{1}|+|\delta_{1}|}\,
   \Gamma\bigl(\mu|\gamma|+\mu|\delta_{1}|\bigr)\,
   \Gamma\bigl(\mu|\alpha-\alpha_{1}|+\sigma|\gamma|+\lambda|\delta-\delta_{1}|
   +\sigma l+q\bigr).
\end{aligned}
\]
Grouping terms and summing over all $k$ with $|k|=m$ gives
\[
\begin{aligned}
\bigl|D_{I}^{\alpha}D_{\varphi}^{\gamma}D_{t}^{\delta}u(\varphi,I;t)\bigr|
&\le d_{0}C_{0}C^{1+l}
   \sum_{m=1}^{\infty}\sum_{|k|=m}|k|^{-l}\Delta(|k|)e^{-C^{-1}|k|^{1/\sigma}}\\
&\quad\times\sum_{0<\alpha_{1}\le\alpha}\sum_{0<\delta_{1}\le\delta}
   C^{\mu|\alpha-\alpha_{1}|+|\gamma|+|\delta-\delta_{1}|}
   C_{0}^{|\alpha_{1}|+|\delta_{1}|} C_{3}^{|\alpha_{1}|+|\delta_{1}|}\\
&\quad\times\frac{\alpha!}{(\alpha-\alpha_{1})!}
   \frac{\delta!}{(\delta-\delta_{1})!}
   \Gamma\bigl(\mu|\alpha_{1}|+\mu|\delta_{1}|\bigr)
   \Gamma\bigl(\mu|\alpha-\alpha_{1}|+\sigma |\gamma|+\lambda|\delta-\delta_{1}|+\sigma l+q\bigr).
\end{aligned}
\]
The number of lattice points $k$ with $|k|=m$ is bounded by $2^{n}m^{n-1}$.
Setting $\varepsilon=C^{-1}C_{0}C_{3}$ and $\varepsilon'=C^{-\mu}C_{0}C_{3}$ we obtain
\[
\begin{aligned}
\bigl|D_{I}^{\alpha}D_{\varphi}^{\gamma}D_{t}^{\delta}u(\varphi,I;t)\bigr|
&\le 2^{n}d_{0}C_{0}C^{1+l}
   \sum_{m=1}^{\infty}m^{n-1-l}\Delta(m)e^{-C^{-1}m^{1/\sigma}}\\
&\quad\times\sum_{0<\alpha_{1}\le\alpha}\varepsilon'^{\,|\alpha_{1}|}
   \sum_{0<\delta_{1}\le\delta}\varepsilon^{\,|\delta_{1}|}
   C^{\mu|\alpha|+|\gamma|+|\delta|}\,
   \Gamma\bigl(\mu|\alpha|+\sigma|\gamma|+\lambda|\delta|+\sigma l+q\bigr).
\end{aligned}
\]
For $\varepsilon,\varepsilon'$ chosen sufficiently small we have
$\sum_{0<\alpha_{1}}\varepsilon'^{\,|\alpha_{1}|}<1$ and $\sum_{0<\delta_{1}}\varepsilon^{\,|\delta_{1}|}<1$.

Applying Lemma~\ref{a7} with $s=n+1-l>0$ and $\eta=C^{-1}>0$ yields
\[
\sup_{Q\ge0}Q^{\,s}\Delta(Q)e^{-C^{-1}Q^{1/\sigma}}
   \le\Gamma_{s}(\eta)\le e^{\eta T^{1/\mu}}.
\]
Since $\sum_{m=1}^{\infty}m^{-2}=\pi^{2}/6$, we finally obtain
\[
\bigl|D_{I}^{\alpha}D_{\varphi}^{\gamma}D_{t}^{\delta}u(\varphi,I;t)\bigr|
   \le d_{0}\mathcal D\,C^{\mu|\alpha|+|\gamma|+|\delta|}
   \Gamma\bigl(\mu|\alpha|+\sigma|\gamma|+\lambda|\delta|+\sigma l+q\bigr),
\]
where
\[
\mathcal D=\frac{\pi^{2}}{6}\,2^{n}\,C_{0}\,C^{1+l}e^{\eta T^{1/\mu}},
\]
and $\eta$ depends linearly on $s$. This completes the proof.
\end{proof}

\section{Proof of Theorem \ref{ad}} \label{sec5}
In this section we first state the relations among the indices
\(\mu,\sigma,\lambda,\varrho\) and \(l\):
\[
\lambda\ge\mu=\rho,\qquad l>\mu,\qquad
\varrho= \sigma l .
\]

We seek symbols \(a\) and \(p^{0}\) in the class
\(\mathcal{S}_{l}\bigl(\mathbb{T}^{n}\times E_{\kappa}(t)\times(-\tfrac12,\tfrac12)\bigr)\)
with asymptotic expansions of the form
\[
a\sim\sum_{j=0}^{\infty}a_{j}(\varphi,I;t)h^{j},\qquad
p^{0}\sim\sum_{j=0}^{\infty}p_{j}^{0}(I;t)h^{j},
\]
where \(a_{j}\in C^{\infty}\bigl(\mathbb{T}^{n}\times E_{\kappa}(t)\times(-\tfrac12,\tfrac12)\bigr)\)
and \(p_{j}^{0}\in C^{\infty}\bigl(E_{\kappa}(t)\times(-\tfrac12,\tfrac12)\bigr)\).

Consider the symbol
\[
c=p\circ a-a\circ p^{0}\sim\sum_{j=0}^{\infty}c_{j}(\varphi,I;t)h^{j}.
\]
From the previous discussion we know
\[
a_{0}(\varphi,I;t)=1,\qquad
p_{0}^{0}(I;t)=p_{0}(I;t)=K_{0}(I;t),
\]
and \(p_{1}^{0}(I;t)=p_{1}(I;t)\), which vanishes on
\(E_{\kappa}(t)\times(-\tfrac12,\tfrac12)\). Consequently
\(c_{0}=c_{1}=0\). For any \(j\ge 2\), the symbolic calculus gives
\begin{equation}\label{d}
c_{j}(\varphi,I;t)=\frac{1}{{\rm i}}(\mathcal{L}_{\omega}a_{j-1})(\varphi,I;t)
+p_{j}(\varphi,I;t)-p_{j}^{0}(I;t)+F_{j}(\varphi,I;t).
\end{equation}
One checks directly that \(F_{2}(\varphi,I;t)=0\), while for \(j\ge 3\)
\[
F_{j}(\varphi,I;t)=F_{j1}(\varphi,I;t)-F_{j2}(\varphi,I;t),
\]
with
\[
\begin{aligned}
F_{j1}(\varphi,I;t)&=
\sum_{s=1}^{j-2}\sum_{r+|\gamma|=j-s}
\frac{1}{\gamma!}\partial_{I}^{\gamma}p_{r}(\varphi,I;t)\,
\partial_{\varphi}^{\gamma}a_{s}(\varphi,I;t),\\[2mm]
F_{j2}(\varphi,I;t)&=
\sum_{s=1}^{j-2}a_{s}(\varphi,I;t)\,p_{j-s}^{0}(I;t).
\end{aligned}
\]
We solve the equations \(c_{j}=0\) for \(j\ge 2\) as follows.
First define
\begin{equation}\label{m}
p_{j}^{0}(I;t)=(2\pi)^{-n}\int_{\mathbb{T}^{n}}
\bigl(p_{j}(\varphi,I;t)+F_{j}(\varphi,I;t)\bigr)\,{\rm d}\varphi,
\end{equation}
and set
\[
f_{j}(\varphi,I;t)=p_{j}^{0}(I;t)-p_{j}(\varphi,I;t)-F_{j}(\varphi,I;t).
\]
Then the homological equation \eqref{d} becomes
\begin{equation}\label{e}
\frac{1}{{\rm i}}(\mathcal{L}_{\omega}a_{j-1})(\varphi,I;t)=f_{j}(\varphi,I;t),
\end{equation}
together with the normalisation condition
\begin{equation}\label{aw}
\int_{\mathbb{T}^{n}}a_{j-1}(\varphi,I;t)\, {\rm d}\varphi=0.
\end{equation}

We now solve equation \eqref{e} under the normalisation condition \eqref{aw} by repeated application of Proposition \ref{av}.

For \(j=2\), we obtain
\[
p_{2}^{0}(I;t)=(2\pi)^{-n}\int_{\mathbb{T}^{n}}p_{2}(\varphi,I;t)\, {\rm d}\varphi,
\]
and the system
\begin{equation}\label{l}
\frac{1}{{\rm i}}(\mathcal{L}_{\omega}a_{1})(\varphi, I;t)=p_{2}^{0}(I;t)-p_{2}(\varphi,I;t),\qquad
\int_{\mathbb{T}^{n}}a_{1}(\varphi,I;t)\,{\rm d}\varphi=0.
\end{equation}
We assume that the symbols \(p_{j}\) (\(j\in\mathbb{Z}_{+}\)) satisfy the estimate
\begin{align}\label{o}
\bigl|\partial_{I}^{\alpha}\partial_{\varphi}^{\beta}\partial_{t}^{\delta}p_{j}(\varphi,I;t)\bigr|
&\leq C_{0}^{\,j+|\alpha|+|\beta|+|\delta|+1}
\,\alpha!^{\mu}\,\beta!^{\sigma}\,\delta!^{\lambda}\,(j!)^{\mu+\sigma+\lambda-1}\\
&\leq C_{0}^{\,j+|\alpha|+|\beta|+|\delta|+1}
\,\alpha!\,\beta!\,\delta!\,
\Gamma_{+}\!\bigl((\mu-1)|\alpha|+(\sigma-1)|\beta|
+(\lambda-1)|\delta| \notag\\
&\qquad\qquad\qquad\qquad\qquad\qquad\qquad
+(\sigma+\mu+\lambda-1)(j-1)\bigr),\notag
\end{align}
for all multi-indices \(\alpha,\beta\in\mathbb{Z}_{+}^{n}\) and all \(\delta,j\in\mathbb{Z}_{+}\).
Here \(\Gamma_{+}(x)=\Gamma(x)\) for \(x\geq1\) and \(\Gamma_{+}(x)=1\) for \(x\leq1\).

Applying Proposition \ref{av} to \eqref{l} yields a solution \(a_{1}\) that fulfills
\[
\bigl|\partial_{I}^{\alpha}\partial_{\varphi}^{\beta}\partial_{t}^{\delta}a_{1}(\varphi,I;t)\bigr|
\leq 2C_{0} \mathcal D\,C^{\mu|\alpha|+|\beta|+|\delta|}
\Gamma\bigl(\mu|\alpha|+\sigma|\beta|+\lambda|\delta|+\varrho\bigr).
\]
Now fix \(j\geq3\) and suppose that we have already constructed functions
\(a_{k}(\varphi,I;t)\) for \(1\leq k\leq j-2\) satisfying \eqref{e} and the bounds
\begin{equation}\label{a2}
\bigl|\partial_{I}^{\alpha}\partial_{\varphi}^{\beta}\partial_{t}^{\delta}a_{k}(\varphi,I;t)\bigr|
\leq d^{\,k}C^{\mu|\alpha|+|\beta|+|\delta|}
\Gamma\bigl(\mu|\alpha|+\sigma|\beta|+\lambda|\delta|+k\varrho\bigr),
\qquad 1\leq k\leq j-2,
\end{equation}
for all \((\varphi,I)\in\mathbb{T}^{n}\times E_{\kappa}(t)\), where the constant
\(d\) is chosen so that \(d\geq 2C_{0} \mathcal D\).

We shall next establish several auxiliary estimates required for the induction step.

\begin{lemma}\label{af}
Assume that \(C \geq 4C_{0}\). Then for any multi-indices \(\alpha, \beta \in \mathbb{Z}_{+}^{n}\) and any nonnegative integer \(\delta\), the following estimate holds for all \((\varphi, I) \in \mathbb{T}^{n} \times E_{\kappa}(t)\):
\[
\bigl| D_{I}^{\alpha} D_{\varphi}^{\beta} D_{t}^{\delta} F_{j1}(\varphi, I; t) \bigr|
\leq Q_{1} d^{\,j-2} C^{\mu|\alpha| + |\beta| + |\delta|}
\Gamma\bigl( \mu|\alpha| + \sigma|\beta| + \lambda|\delta| + (j-1)\varrho \bigr),
\]
where
\[
Q_{1} = \frac{R_{0}C_{0}}{\delta} \sum_{p \in \mathbb{Z}_{+}^{n+1}}
\frac{C^{2|p|+2}}{\Gamma\bigl(\delta(|p|+1)\bigr)} < \infty.
\]
\end{lemma}

\begin{proof}
For notational convenience, define
\begin{align}\label{a8}
S_{r,s,\gamma}(\varphi,I;t)=\frac{1}{\gamma!}\,
\partial_{I}^{\gamma}p_{r}(\varphi,I;t)\;
\partial_{\varphi}^{\gamma}a_{s}(\varphi,I;t),
\end{align}
where the indices satisfy
\begin{equation}\label{bj}
3\le r+s+|\gamma|=j,\qquad 1\le s\le j-2 .
\end{equation}
From \eqref{a8} together with the bounds \eqref{o} and \eqref{a2}, we obtain
\begin{align*}
\bigl| \partial_{I}^{\alpha}\partial_{\varphi}^{\beta}\partial_{t}^{\delta}
      S_{r,s,\gamma}(\varphi,I;t) \bigr|
&\le\sum_{\alpha_{1}\le\alpha}\sum_{\beta_{1}\le\beta}\sum_{\delta_{1}\le\delta}
   \frac{1}{\gamma!}\binom{\alpha}{\alpha_{1}}\binom{\beta}{\beta_{1}}
   \binom{\delta}{\delta_{1}} \\
&\qquad\times
   \bigl|\partial_{I}^{\gamma+\alpha_{1}}\partial_{\varphi}^{\beta_{1}}
   \partial_{t}^{\delta_{1}}p_{r}(\varphi,I;t)\bigr| \\
&\qquad\times
   \bigl|\partial_{I}^{\alpha-\alpha_{1}}\partial_{\varphi}^{\gamma+\beta-\beta_{1}}
   \partial_{t}^{\delta-\delta_{1}}a_{s}(\varphi,I;t)\bigr| \\[2mm]
&\le d^{\,s}\sum_{\alpha_{1}\le\alpha}\sum_{\beta_{1}\le\beta}
   \sum_{\delta_{1}\le\delta}
   \frac{(\gamma+\alpha_{1})!\,\beta_{1}!\,\delta_{1}!}{\gamma!}
   \binom{\alpha}{\alpha_{1}}\binom{\beta}{\beta_{1}}\binom{\delta}{\delta_{1}} \\
&\qquad\times
   \Gamma\bigl((\mu-1)|\gamma+\alpha_{1}|+(\sigma-1)|\beta_{1}|
   +(\lambda-1)|\delta_{1}|+(\mu+\sigma+\lambda-1)(r-1)\bigr) \\
&\qquad\times
   \Gamma\bigl(\mu|\alpha-\alpha_{1}|+\sigma|\gamma+\beta-\beta_{1}|
   +\lambda|\delta-\delta_{1}|+s\varrho\bigr) \\
&\qquad\times
   C_{0}^{|\gamma+\alpha_{1}|+|\beta_{1}|+|\delta_{1}|+r+1}\;
   C^{\mu|\alpha-\alpha_{1}|+|\gamma+\beta-\beta_{1}|+|\delta-\delta_{1}|}.
\end{align*}
Applying Lemma~\ref{bh} yields
\begin{align*}
\bigl| \partial_{I}^{\alpha}\partial_{\varphi}^{\beta}\partial_{t}^{\delta}
      S_{r,s,\gamma}(\varphi,I;t) \bigr|
&\le d^{\,s}C^{\mu|\alpha|+|\beta|+|\delta|}\,
   \Gamma\bigl(\mu|\alpha|+\sigma|\beta|+\lambda|\delta|
   +(\sigma+\mu+\lambda-1)(|\gamma|+r-1)+s\varrho\bigr)\\
&\quad\times\sum_{\alpha_{1}\le\alpha}\sum_{\beta_{1}\le\beta}
   \sum_{\delta_{1}\le\delta}
   (2C_{0}/C)^{|\alpha_{1}+\beta_{1}|+|\delta_{1}|}\,
   (2C)^{|\gamma|}\,C_{0}^{|\gamma|+r+1}.
\end{align*}
Set \(\delta_0=\varrho-\sigma-\mu-\lambda+1\).  Because
\(\varrho=\sigma l, l>\mu\), we have
\[
\delta_0>\sigma\mu-\mu-\sigma+1=(1-\sigma)(1-\mu)>0 .
\]
Moreover, note the elementary inequality
\begin{equation}\label{bm}
(\sigma+\mu+\lambda-1)(|\gamma|+r-1)\ge\sigma+\mu+\lambda-1\ge1 .
\end{equation}
Consequently,
\begin{align*}
(j-1)\varrho-\delta_0(|\gamma|+r-1)
&= (|\gamma|+r-1+s)\varrho-\delta_0(|\gamma|+r-1) \\
&= (\sigma+\mu+\lambda-1)(|\gamma|+r-1)+s\varrho\ge1,
\end{align*}
since \(r+s+|\gamma|=j\) and \(|\gamma|+r\ge2\).

Using Lemma~\ref{bi} we obtain
\begin{align*}
&\Gamma\!\bigl(\mu|\alpha|+\sigma|\beta|+\lambda|\delta|
   +(j-1)\varrho-\delta_0(|\gamma|+r-1)\bigr)\,
   \Gamma\!\bigl(\delta_0(|\gamma|+r-1)\bigr)\\
&\qquad\le\frac{1}{\delta_0(|\gamma|+r-1)}\,
       \Gamma\!\bigl(\mu|\alpha|+\sigma|\beta|+\lambda|\delta|
       +(j-1)\varrho\bigr).
\end{align*}
Hence
\begin{align*}
&\Gamma\!\bigl(\mu|\alpha|+\sigma|\beta|+\lambda|\delta|
   +(\sigma+\mu+\lambda-1)(|\gamma|+r-1)+s\varrho\bigr)\\
&\qquad=\Gamma\!\bigl(\mu|\alpha|+\sigma|\beta|+\lambda|\delta|
   +(j-1)\varrho-\delta_0(|\gamma|+r-1)\bigr)\\
&\qquad\le\frac{1}{\delta_0}\,
       \Gamma\!\bigl(\mu|\alpha|+\sigma|\beta|+\lambda|\delta|
       +(j-1)\varrho\bigr)\,
       \bigl(\Gamma(\delta_0(|\gamma|+r-1))\bigr)^{-1}.
\end{align*}

Now assume \(C>4C_{0}\).  For any admissible triple \((r,s,\gamma)\) satisfying
\eqref{bj} we therefore obtain
\begin{align*}
\bigl| \partial_{I}^{\alpha}\partial_{\varphi}^{\beta}\partial_{t}^{\delta}
      S_{r,s,\gamma}(\varphi,I;t) \bigr|
&\le Q_{0}\,d^{\,j-2}C^{\mu|\alpha|+|\beta|+|\delta|}
   \,C_{0}\,C^{2(|\gamma|+r)}\\
&\quad\times\Gamma\!\bigl(\mu|\alpha|+\sigma|\beta|+\lambda|\delta|
   +(j-1)\varrho\bigr)
   \bigl(\delta_0\,\Gamma(\delta_0(|\gamma|+r-1))\bigr)^{-1},
\end{align*}
where \(Q_{0}^{1/2}=\sum_{\alpha\in\mathbb{Z}_{+}^{n+1}}2^{-|\alpha|}\).

Finally, because
\[
F_{j1}(\varphi,I;t)=\sum_{s=1}^{j-2}\;
\sum_{\substack{r+|\gamma|=j-s\\ r\ge1,\;\gamma\ge0}}
\frac{1}{\gamma!}\,
\partial_{I}^{\gamma}p_{r}(\varphi,I;t)\;
\partial_{\varphi}^{\gamma}a_{s}(\varphi,I;t),
\]
we conclude that
\begin{align*}
\bigl| D_{I}^{\alpha}D_{\varphi}^{\beta}D_{t}^{\delta}
      F_{j1}(\varphi,I;t) \bigr|
&\le\sum_{s=1}^{j-2}\;
   \sum_{\substack{r+|\gamma|=j-s\\ r\ge1,\;\gamma\ge0}}
   \bigl| D_{I}^{\alpha}D_{\varphi}^{\beta}D_{t}^{\delta}
         S_{r,s,\gamma}(\varphi,I;t) \bigr|\\
&\le Q_{1}\,d^{\,j-2}C^{\mu|\alpha|+|\beta|+|\delta|}\,
   \Gamma\!\bigl(\mu|\alpha|+\sigma|\beta|+\lambda|\delta|
   +(j-1)\varrho\bigr),
\end{align*}
with
\[
Q_{1}= \frac{Q_{0}C_{0}}{\delta_0}
       \sum_{p\in\mathbb{Z}_{+}^{n+1}}
       \frac{C^{2(|p|+1)}}{\Gamma(\delta_0(|p|+1))}
       \;<\;\infty .
\]
This completes the proof of Lemma~\ref{af}.

\end{proof}
We now prove, based on Lemma~\ref{af}, the existence of functions
$p_{k}^{0}(I;t)$ for $2 \leq k \leq j-1$ satisfying~\eqref{m} such that
\begin{align}\label{n}
\bigl|\partial_{I}^{\alpha}\partial_{t}^{\delta}p_{k}^{0}(I;t)\bigr|
&\leq d^{\,k-3/2}\, C^{\mu|\alpha|+|\delta|}\,
\Gamma\!\bigl(\mu|\alpha|+\lambda|\delta|+(k-1)\varrho\bigr),
\qquad 2\leq k\leq j-1 .
\end{align}
Given the recursive definition~\eqref{m}, the estimate for
$p_{j}^{0}(I;t)$ with $j\geq 3$ follows from the known estimates for
$p_{j}(\varphi,I;t)$ and $F_{j}(\varphi,I;t)$. Note that~\eqref{e} implies
\[
\int_{\mathbb{T}^{n}} F_{j2}(\varphi,I;t)\,\mathrm{d}\varphi
= \sum_{s=1}^{j-2} p_{j}^{0}(I;t) \int_{\mathbb{T}^{n}}
a_{s}(\varphi,I;t)\,\mathrm{d}\varphi = 0,
\]
from which we obtain
\[
p_{j}^{0}(I;t) = (2\pi)^{-n} \int_{\mathbb{T}^{n}}
\bigl(p_{j}(\varphi,I;t)+F_{j1}(\varphi,I;t)\bigr)\,\mathrm{d}\varphi .
\]
Applying the bounds from~\eqref{o} and Lemma~\ref{af} yields, for any
$j \geq 2$, the estimate
\begin{align}\label{p}
\bigl|\partial_{I}^{\alpha}\partial_{t}^{\delta}p_{k}^{0}(I;t)\bigr|
&\leq Q_{1}\, d^{\,k-2}\, C^{\mu|\alpha|+|\delta|}\,
\Gamma\!\bigl(\mu|\alpha|+\lambda|\delta|+(k-1)\varrho\bigr) \notag \\
&\qquad + C_{0}^{\,k+|\alpha|+|\delta|+1}\,
\Gamma_{+}\!\bigl(\mu|\alpha|+\lambda|\delta|
+(\mu+\lambda-1)(k-1)\bigr) \notag\\[2mm]
&\leq d^{\,k-3/2}\, C^{\mu|\alpha|+|\delta|}\,
\Gamma\!\bigl(\mu|\alpha|+\lambda|\delta|+(k-1)\varrho\bigr).
\end{align}
The final inequality uses the condition $\varrho > \mu+\sigma+\lambda-1$.
Here the constant $d$ is taken sufficiently large, depending only on
$n$, $\mu$, $\sigma$, $C_{0}$, and $C$. This completes the inductive
proof of~\eqref{n}.

\begin{lemma}\label{ag}
For all multi-indices \(\alpha, \beta \in \mathbb{Z}_{+}^{n}\) and every \(\delta \in \mathbb{Z}_{+}\), the function \(F_{j2}\) satisfies the estimate
\begin{align*}
\bigl|D_{I}^{\alpha}D_{\varphi}^{\beta}D_{t}^{\delta}F_{j2}(\varphi,I;t)\bigr|
&\leq M_{2}\,d^{\,j-3/2}\,C^{\mu|\alpha|+|\beta|+|\delta|}
      \,\Gamma\bigl(\mu|\alpha|+\sigma|\beta|+\lambda|\delta|+(j-1)\varrho\bigr),\\
&\qquad\qquad (\varphi,I)\in \mathbb{T}^{n}\times E_{\kappa}(t),
\end{align*}
where
\[
M_{2}=2M\Bigl(\sum_{\alpha_{1}\in\mathbb{Z}^{n}_{+}}2^{-|\alpha_{1}|/6}\Bigr)
      \Bigl(\sum_{\delta_{1}\in\mathbb{Z}_{+}}2^{-|\delta_{1}|/6}\Bigr)
      \sum_{s=1}^{j-2}\binom{j-2}{s-1}^{-1/3}.
\]

\end{lemma}
\begin{proof}
From \eqref{n} and \eqref{a2} we obtain, for \(1\leq s\leq j-2\),
\begin{align*}
&\bigl|D_{I}^{\alpha}D_{\varphi}^{\beta}D_{t}^{\delta}
      \bigl(a_{s}(\varphi,I;t)p_{j-s}^{0}(I;t)\bigr)\bigr| \\
&\leq\sum_{\alpha_{1}\leq\alpha}\sum_{\delta_{1}\leq\delta}
      \binom{\alpha}{\alpha_{1}}\binom{\delta}{\delta_{1}}
      \bigl|D_{I}^{\alpha_{1}}D_{\varphi}^{\beta}D_{t}^{\delta_{1}}
      a_{s}(\varphi,I;t)\bigr|\,
      \bigl|D_{I}^{\alpha-\alpha_{1}}D_{t}^{\delta-\delta_{1}}
      p_{j-s}^{0}(I;t)\bigr| \\
&\leq d^{\,j-3/2}C^{\mu|\alpha|+|\beta|+|\delta|}
      \sum_{\alpha_{1}\leq\alpha}\sum_{\delta_{1}\leq\delta}
      \binom{\alpha}{\alpha_{1}}\binom{\delta}{\delta_{1}}
      \Gamma\bigl(\mu|\alpha_{1}|+\sigma|\beta|+\lambda|\delta_{1}|+s\varrho\bigr) \\
&\qquad\qquad\qquad\qquad\qquad\qquad\times
      \Gamma\bigl(\mu|\alpha-\alpha_{1}|+\lambda|\delta-\delta_{1}|
      +(j-s-1)\varrho\bigr).
\end{align*}
Since Lemma~\ref{bw} together with the inequalities
\[
B\bigl(s\varrho,\,(j-s-1)\varrho\bigr)
   <B(s,j-s-1)<\binom{j-2}{s-1}^{-1}
\]
implies
\begin{align*}
&\bigl|D_{I}^{\alpha}D_{\varphi}^{\beta}D_{t}^{\delta}
      \bigl(a_{s}(\varphi,I;t)p_{j-s}^{0}(I;t)\bigr)\bigr| \\
&\leq M\,d^{\,j-3/2}C^{\mu|\alpha|+|\beta|+|\delta|}
      \sum_{\alpha_{1}\leq\alpha}\sum_{\delta_{1}\leq\delta}
      \binom{|\alpha|}{|\alpha_{1}|}^{-1/6}
      \binom{|\delta|}{|\delta_{1}|}^{-1/6}
      B\bigl(s\varrho,\,(j-s-1)\varrho\bigr)^{1/3} \\
&\qquad\qquad\qquad\qquad\qquad\qquad\times
      \Gamma\bigl(\mu|\alpha|+\sigma|\beta|+\lambda|\delta|+(j-1)\varrho\bigr) \\
&\leq M_{1}\,d^{\,j-3/2}C^{\mu|\alpha|+|\beta|+|\delta|}
      \binom{j-2}{s-1}^{-1/3}
      \Gamma\bigl(\mu|\alpha|+\sigma|\beta|+\lambda|\delta|+(j-1)\varrho\bigr),
\end{align*}
where
\(M_{1}=2M\bigl(\sum_{\alpha_{1}\in\mathbb{Z}_{+}^{n}}2^{-|\alpha_{1}|/6}\bigr)
        \bigl(\sum_{\delta_{1}\in\mathbb{Z}_{+}}2^{-|\delta_{1}|/6}\bigr)\).

Observing that
\[
\sum_{s=1}^{j-2}\binom{j-2}{s-1}^{-1/3}
   \leq 2\sum_{p=0}^{\infty}2^{-p/3}<\infty,
\]
we sum over \(s=1,\dots,j-2\) to obtain
\[
\bigl|D_{I}^{\alpha}D_{\varphi}^{\beta}D_{t}^{\delta}F_{j2}(\varphi,I;t)\bigr|
\leq M_{2}\,d^{\,j-3/2}C^{\mu|\alpha|+|\beta|+|\delta|}
      \Gamma\bigl(\mu|\alpha|+\sigma|\beta|+\lambda|\delta|+(j-1)\varrho\bigr),
\]
which completes the proof.

\end{proof}
Finally, combining Lemma~\ref{af}, Lemma~\ref{ag}, and estimates~\eqref{o}
and~\eqref{p}, we obtain for every \(j \geq 3\) the following bound on the
right-hand side of the homological equation~\eqref{e}:
\begin{align*}
\bigl|D_{I}^{\alpha}D_{\varphi}^{\beta}D_{t}^{\delta} f_{j}(\varphi,I;t)\bigr|
&\leq M_{3} \, d^{\,j-2/3} \, C^{\mu|\alpha|+|\beta|+|\delta|}
      \Gamma\bigl(\mu|\alpha|+\sigma|\beta|+\lambda|\delta|+(j-1)\varrho\bigr), \\
&\qquad \forall\,\alpha,\beta\in\mathbb{Z}_{+}^{n},\; \delta\in\mathbb{Z}_{+},
\end{align*}
where \(M_{3}\) depends only on \(n\), \(\mu\), \(\sigma\), \(\lambda\), \(C_{0}\),
and \(C\).
Applying Proposition~\ref{av} then yields a solution \(a_{j-1}\) of
\eqref{e} and \eqref{aw} that satisfies \eqref{a2} for \(k=j-1\),
provided the parameter \(d\) is chosen sufficiently large.

\appendix  % ====== 插入附录的起始点 ======

\section{proof of the  estimate (\ref{c})}\label{A1}

From \eqref{f} and \eqref{bx}, the function $f_k(I;t)$ belongs to the Gevrey class
$\mathcal{G}^{\mu,\lambda}(E_{\kappa}(t)\times(-\tfrac{1}{2},\tfrac{1}{2}))$.
Using integration by parts and the results of \cite{MR1055164}, we obtain the estimate
\begin{align}\label{a11}
\bigl|\partial_{I}^{\alpha}\partial_{t}^{\delta}f_{k}(I;t)\bigr|
\leq d_{0}\,C^{\mu|\alpha|+|\delta|+N+1}\,
\Gamma\bigl(\mu|\alpha|+\lambda|\delta|+q\bigr)\,
N!^{\sigma}\,|k|^{-N}.
\end{align}
We minimize the right-hand side of \eqref{a11} with respect to $N\in\mathbb{N}_{+}$
by means of Stirling's formula.  The optimal choice of $N$ is approximately
\begin{align}\label{a10}
N\sim\Bigl(\frac{|k|}{C}\Bigr)^{\!1/\sigma},
\end{align}
which leads to
\begin{align*}
\bigl|\partial_{I}^{\alpha}\partial_{t}^{\delta}f_{k}(I;t)\bigr|
\leq d_{0}\,C^{\mu|\alpha|+|\delta|+1}\,
\Gamma\bigl(\mu|\alpha|+\lambda|\delta|+q\bigr)\,
\exp\!\bigl(-C^{-1}|k|^{1/\sigma}\bigr).
\end{align*}

To see why \eqref{a10} gives the optimum, consider the part of the bound in
\eqref{a11} that depends on $N$, namely
\[
f(N):=C^{N}N!^{\sigma}|k|^{-N}.
\]
Taking logarithms and applying Stirling's approximation
$\log N! = N\log N - N + O(\log N)$, we obtain
\begin{align*}
\log f(N) & = N\log C + \sigma\bigl(N\log N - N\bigr) - N\log|k| + O(\log N)\\
       &  = \sigma N\log N + N\bigl(\log C - \log|k| - \sigma\bigr) + O(\log N).
\end{align*}
Treating $N$ as a continuous variable and differentiating, we have
\[
\frac{\rm d}{{\rm d}N}\log f(N) \approx \sigma\log N + \sigma + \log C - \log|k| - \sigma
                         = \sigma\log N + \log C - \log|k|.
\]
Setting the derivative to zero yields $\sigma\log N \approx \log|k| - \log C$,
hence $N \approx (|k|/C)^{1/\sigma}$, which is precisely \eqref{a10}.

\section{Gevrey function class}\label{secB1}
This section defines the Gevrey function spaces used throughout the paper and recalls their basic properties.

For a parameter \(\mu \geq 1\), the Gevrey class of order \(\mu\) consists of smooth functions whose derivatives grow at most like \(s^{-|k|}k!^{\mu}\) for some \(s>0\); equivalently, their Fourier coefficients (when defined on a torus) decay as \(e^{-\mu s|k|^{1/\mu}}\). The case \(\mu = 1\) corresponds precisely to the space of real-analytic functions. These classes were introduced by M. Gevrey in the early twentieth century in the context of partial differential equations and have since found applications in control theory, mathematical physics, and signal processing.
\begin{definition}
Let \(\mu \geq 1\) and let \(X \subset \mathbb{R}^{n}\) be an open set. For a constant \(L>0\) we define the Gevrey space of order \(\mu\) and radius \(L\) as
\[
\mathcal{G}^{\mu,L}(X)=\Bigl\{ f\in C^{\infty}(X) :
\|f\|_{\mu,L}:=\sup_{\alpha\in\mathbb{N}^{n}}\sup_{x\in X}
|\partial_{x}^{\alpha}f(x)|\,L^{-|\alpha|}\,(\alpha!)^{-\mu}<\infty\Bigr\}.
\]
The Gevrey class of order \(\mu\) on \(X\) is then
\[
\mathcal{G}^{\mu}(X)=\bigcup_{L>0}\mathcal{G}^{\mu,L}(X).
\]
\end{definition}
An equivalent characterization, given in \cite{MR4050197}, is the following:
\begin{definition}
Let \(L>0\). A function \(H\) that is smooth on an open neighbourhood of
\(\mathbb{T}^{n}\times D\) is called a Gevrey-\((\alpha,L)\) function, denoted
\(H\in \mathcal{G}^{\alpha,L}(\mathbb{T}^{n}\times D)\), if there exists
\(s_{0}>0\) such that
\[
|H|_{\alpha,L}:=c\sup_{(\theta,I)\in\mathbb{T}^{n}\times D}
\Bigl(\sup_{k\in\mathbb{N}^{2n}}
\frac{(|k|+1)^{2}L^{\alpha|k|}\partial^{k}H(\theta,I)}{|k|!^{\alpha}}\Bigr)<\infty,
\qquad c:=\frac{4\pi^{2}}{3}.
\]
\end{definition}
The norm \(|\cdot|_{\alpha,L}\) makes \(\mathcal{G}^{\alpha,L}(\mathbb{T}^{n}\times D)\) a Banach algebra, a fact that is particularly convenient for nonlinear estimates. For the analysis in this paper we require anisotropic Gevrey classes, which we define as follows.

\begin{definition}\label{at}
Let \(\sigma,\mu\geq1\) and \(L_{1},L_{2}>0\). A function
\(f\in C^{\infty}(\mathbb{T}^{n}\times D)\) belongs to the anisotropic Gevrey class
\(\mathcal{G}^{\sigma,\mu}_{L_{1},L_{2}}(\mathbb{T}^{n}\times D)\) if
\[
\sup_{\alpha,\beta\in\mathbb{N}^{n}}
\sup_{(\varphi,I)\in\mathbb{T}^{n}\times D}
|\partial_{\varphi}^{\alpha}\partial_{I}^{\beta}f(\varphi,I)|\,
L_{1}^{-|\alpha|}L_{2}^{-|\beta|}\,
(\alpha!)^{-\sigma}(\beta!)^{-\mu} < \infty .
\]

\end{definition}
A similar result for three variables is recorded, for instance, in \cite{MR2104602}.
We adapt it to the anisotropic setting as follows.

\begin{definition}\label{a12}
Let \(\rho\geq 1\) and \(L_{1},L_{2}>0\). A function
\(f\in C^{\infty}\bigl(\mathbb{T}^{n}\times D\times (-\tfrac12,\tfrac12)\bigr)\) belongs to the anisotropic Gevrey class
\(\mathcal{G}^{\rho,\rho+1,\rho+1}_{L_{1},L_{2},L_{2}}\bigl(\mathbb{T}^{n}\times D\times (-\tfrac12,\tfrac12)\bigr)\) if
\[
\sup_{\alpha,\beta,\delta}\sup_{(\varphi,I,t)}
|\partial_{\varphi}^{\alpha}\partial_{I}^{\beta}\partial_{t}^{\delta}f(\varphi,I;t)|\,
L_{1}^{-|\alpha|}L_{2}^{-|\beta|-|\delta|}\,
(\alpha!)^{-\rho}(\beta!)^{-(\rho+1)}(\delta!)^{-(\rho+1)}<\infty.
\]

\end{definition}
\section{Gamma function}\label{secCb 1}
We recall some basic properties of the Gamma function, referring to \cite{MR1770800} for details.
For \(x>0\), the Gamma function is defined as
\[
\Gamma(x)=\int_{0}^{\infty}e^{-t}t^{\,x-1}\,\mathrm{d}t.
\]
It satisfies the identity
\begin{align}\label{bn}
\Gamma(x)\Gamma(y)=\Gamma(x+y)B(x,y),\qquad x,y>0,
\end{align}
where the Beta function is given by
\[
B(x,y)=\int_{0}^{1}(1-t)^{x-1}t^{y-1}\,\mathrm{d}t.
\]
Note that for any \(x\ge 1\) and \(y>0\) one has \(B(x,y)\le y^{-1}\).
Combining this bound with \eqref{bn} yields the following estimate.

\begin{lemma}\label{bi}
For all \(x\ge 1\) and \(y>0\),
\[
\Gamma(x)\Gamma(y)\le\frac{1}{y}\,\Gamma(x+y).
\]
\end{lemma}
For multi-indices \(x,y\in\mathbb{Z}_{+}^{n}\) with \(0\le y\le x\) we write
\[
\binom{x}{y}=\frac{x!}{y!\,(x-y)!},
\]
with the convention \(0!=1\).

\begin{lemma}\label{bh}

Let \(\alpha,\beta,\delta\in\mathbb{Z}_{+}^{n}\) be multi-indices, and let
\(\alpha_{1}\le\alpha\), \(\beta_{1}\le\beta\), \(\delta_{1}\le\delta\) and
\(\gamma\in\mathbb{Z}_{+}^{n}\). Assume further that \(s\geq1\), \(r\geq0\) and
that \(|\gamma|+r\geq2\). Then
\begin{multline*}
(\gamma+\alpha_{1})!\,
\frac{\beta_{1}!\,\delta_{1}!}{\gamma!}\,
\binom{\alpha}{\alpha_{1}}\binom{\beta}{\beta_{1}}\binom{\delta}{\delta_{1}}
\;\Gamma\!\bigl((\mu-1)|\gamma+\alpha_{1}|
      +(\sigma-1)|\beta_{1}|+(\lambda-1)|\delta_{1}|
      +(\sigma+\mu+\lambda-1)(r-1)\bigr)\\
\times\Gamma\!\bigl(\mu|\alpha-\alpha_{1}|+\sigma|\gamma+\beta-\beta_{1}|
      +\lambda|\delta-\delta_{1}|+s\bar{\rho}\bigr)\\[2mm]
\leq 2^{\,|\gamma+\alpha_{1}|}\,
\Gamma\!\bigl(\mu|\alpha|+\sigma|\beta|+\lambda|\delta|
      +(\sigma+\mu+\lambda-1)(|\gamma|+r-1)+s\bar{\rho}\bigr).
\end{multline*}

\end{lemma}
\begin{proof}
We first note the elementary identities
\( x\Gamma(x)=\Gamma(x+1) \) for \(x>0\) and the factorial estimate
\[
(|\alpha|+k)!\le 2^{\,|\alpha|+k}|\alpha|!\,k!.
\]

Applying these together with the definition of the multinomial coefficients, we obtain
\begin{align*}
&(\gamma+\alpha_{1})!\,
  \frac{\beta_{1}!\,\delta_{1}!}{\gamma!}\,
  \binom{\alpha}{\alpha_{1}}\binom{\beta}{\beta_{1}}\binom{\delta}{\delta_{1}}
  \,\Gamma\!\bigl(\mu|\alpha-\alpha_{1}|+\sigma|\gamma+\beta-\beta_{1}|
                +\lambda|\delta-\delta_{1}|+s\bar{\rho}\bigr)\\[1mm]
&\leq 2^{\,|\gamma+\alpha_{1}|}\,
      \frac{|\alpha|!}{|\alpha-\alpha_{1}|!}\,
      \frac{|\beta|!}{|\beta-\beta_{1}|!}\,
      \frac{|\delta|!}{|\delta-\delta_{1}|!}
      \,\Gamma\!\bigl(\mu|\alpha-\alpha_{1}|+\sigma|\beta-\beta_{1}|
                     +\sigma|\gamma|+\lambda|\delta-\delta_{1}|+s\bar{\rho}\bigr)\\[1mm]
&\leq 2^{\,|\gamma+\alpha_{1}|}\,
      \Gamma\!\bigl(|\alpha|+|\beta|+|\delta|
                     +(\mu-1)|\alpha-\alpha_{1}|+(\sigma-1)|\beta-\beta_{1}|
                     +\sigma|\gamma|+(\lambda-1)|\delta-\delta_{1}|+s\bar{\rho}\bigr)\\[1mm]
&= 2^{\,|\gamma+\alpha_{1}|}\,
   \Gamma\!\bigl(\mu|\alpha|+\sigma|\beta|+\lambda|\delta|
                 -(\mu-1)|\alpha_{1}|-(\sigma-1)|\beta_{1}|
                 -(\lambda-1)|\delta_{1}|+\sigma|\gamma|+s\bar{\rho}\bigr).
\end{align*}

On the other hand, since \(s\bar{\rho}>1\) and by \eqref{bm}, we have
\[
(\mu-1)|\gamma+\alpha_{1}|+(\mu+\sigma+\lambda-1)(r-1) > 1.
\]
Applying Lemma~\ref{bi} then yields
\begin{multline*}
(\gamma+\alpha_{1})!\,\frac{\beta_{1}!\,\delta_{1}!}{\gamma!}\,
\binom{\alpha}{\alpha_{1}}\binom{\beta}{\beta_{1}}\binom{\delta}{\delta_{1}}
\;\Gamma\!\bigl((\mu-1)|\gamma+\alpha_{1}|+(\sigma-1)|\beta_{1}|+(\lambda-1)|\delta_{1}|
      +(\sigma+\mu+\lambda-1)(r-1)\bigr)\\
\times\Gamma\!\bigl(\mu|\alpha-\alpha_{1}|+\sigma|\gamma+\beta-\beta_{1}|
      +\lambda|\delta-\delta_{1}|+s\bar{\rho}\bigr)\\
\leq 2^{\,|\gamma+\alpha_{1}|}\,
\Gamma\!\bigl(\mu|\alpha|+\sigma|\beta|+\lambda|\delta|
      +(\sigma+\mu+\lambda-1)(|\gamma|+r-1)+s\bar{\rho}\bigr),
\end{multline*}
provided that \(\rho\geq\sigma\). This completes the proof.

\end{proof}

\begin{lemma}\label{bw}
Assume \(\rho \geq 7\). Then there exists a constant \(M>0\) such that for all
non-negative integers \(x_1,y_1,x_2,y_2\) and all real numbers \(p\geq1,\;q\geq1\),
\begin{multline*}
\binom{x_1+y_1}{x_1}^{7/6}\binom{x_2+y_2}{x_2}^{7/6}
\;\Gamma\!\bigl(\mu x_1+\lambda x_2+p\bigr)\,
\Gamma\!\bigl(\mu y_1+\lambda y_2+q\bigr)\\[2mm]
\leq M\;\Gamma\!\bigl(\mu(x_1+y_1)+\lambda(x_2+y_2)+p+q\bigr)\,
B(p,q)^{1/3},
\end{multline*}
where \(B(p,q)\) denotes the Beta function.

\end{lemma}
\begin{proof}
For \(x_1,y_1,x_2,y_2\ge 1\), we start from the identity \eqref{bn} to obtain
\begin{align}\label{bo}
&\Gamma(\mu x_1+\lambda x_2+p)\,\Gamma(\mu y_1+\lambda y_2+q) \notag\\
&\qquad \le \Gamma\bigl(\mu(x_1+y_1)+\lambda(x_2+y_2)+p+q\bigr)\,
           B\bigl(\mu x_1+\lambda x_2+p,\;\mu y_1+\lambda y_2+q\bigr).
\end{align}
From the definition of the Beta function,
\[
B(\mu x_1+\lambda x_2+p,\mu y_1+\lambda y_2+q)
   =\int_0^1 t^{\mu x_1+\lambda x_2+p-1}(1-t)^{\mu y_1+\lambda y_2+q-1}\,\mathrm{d}t,
\]
we immediately have the three estimates
\begin{align}
B(\mu x_1+\lambda x_2+p,\mu y_1+\lambda y_2+q) &\le B(\mu x_1,\mu y_1),\label{bp}\\[1mm]
B(\mu x_1+\lambda x_2+p,\mu y_1+\lambda y_2+q) &\le B(\lambda x_2,\lambda y_2),\label{tt}\\[1mm]
B(\mu x_1+\lambda x_2+p,\mu y_1+\lambda y_2+q) &\le B(p,q).\label{bq}
\end{align}
Inserting \eqref{bp}--\eqref{bq} into \eqref{bo} gives
\begin{multline*}
\Gamma(\mu x_1+\lambda x_2+p)\,\Gamma(\mu y_1+\lambda y_2+q)\\
\le \Gamma\bigl(\mu(x_1+y_1)+\lambda(x_2+y_2)+p+q\bigr)\,
    B(\mu x_1,\mu y_1)^{1/3}B(\lambda x_2,\lambda y_2)^{1/3}B(p,q)^{1/3}.
\end{multline*}
Inserting \eqref{bp}--\eqref{bq} into \eqref{bo} gives
\begin{multline*}
\Gamma(\mu x_1+\lambda x_2+p)\,\Gamma(\mu y_1+\lambda y_2+q)\\
\le \Gamma\bigl(\mu(x_1+y_1)+\lambda(x_2+y_2)+p+q\bigr)\,
    B(\mu x_1,\mu y_1)^{1/3}B(\lambda x_2,\lambda y_2)^{1/3}B(p,q)^{1/3}.
\end{multline*}

It therefore suffices to prove that for a suitable constant \(M>0\),
\begin{align}\label{keyest}
B(\mu x_1,\mu y_1) &\le M^{1/2}\binom{x_1+y_1}{x_1}^{-7/2},\\[1mm]
B(\lambda x_2,\lambda y_2) &\le M^{1/2}\binom{x_2+y_2}{x_2}^{-7/2}.
\end{align}
We only show the first inequality; the second is completely analogous.

By Stirling's formula there exists a constant \(L>1\) such that for every integer \(n\ge 1\),
\[
L^{-1}\le \Gamma(n)\,(2\pi)^{-1/2}n^{\,1/2-n}e^{\,n}\le L.
\]
Applying this to \(\Gamma(\mu x_1)\), \(\Gamma(\mu y_1)\) and \(\Gamma(\mu(x_1+y_1))\) yields
\begin{align*}
\Gamma(\mu x_1) &\le L^{\mu+1}\,\Gamma(x_1)^{\mu}
                 \Bigl(\frac{x_1}{2\pi}\Bigr)^{\!\frac{\mu-1}{2}}
                 \mu^{\mu x_1-1/2},\\[1mm]
\Gamma(\mu y_1) &\le L^{\mu+1}\,\Gamma(y_1)^{\mu}
                 \Bigl(\frac{y_1}{2\pi}\Bigr)^{\!\frac{\mu-1}{2}}
                 \mu^{\mu y_1-1/2},\\[1mm]
\Gamma(\mu(x_1+y_1))^{-1}&\le L^{\mu+1}\,\Gamma(x_1+y_1)^{-\mu}
                 \Bigl(\frac{x_1+y_1}{2\pi}\Bigr)^{\!\frac{1-\mu}{2}}
                 \mu^{\,1/2-\mu(x_1+y_1)}.
\end{align*}
Using again \eqref{bn} we obtain
\begin{align*}
B(\mu x_1,\mu y_1)
&=\frac{\Gamma(\mu x_1)\Gamma(\mu y_1)}{\Gamma(\mu(x_1+y_1))}\\[1mm]
&\le L^{3(\mu+1)}
   \Bigl(\frac{1}{2\pi}\Bigr)^{\!\frac{1-\mu}{2}}
   B(x_1,y_1)^{\mu}
   \Bigl(\frac{x_1+y_1}{x_1y_1}\Bigr)^{\!\frac{1-\mu}{2}}
   \mu^{-1/2}\\[1mm]
&\le M^{1/2}
   \Bigl(\frac{x_1y_1}{x_1+y_1}\Bigr)^{\!\frac{\mu-1}{2}}
   B(x_1,y_1)^{\frac{\mu-1}{2}}
   = M^{1/2}\binom{x_1+y_1}{x_1}^{\!-\frac{\mu-1}{2}},
\end{align*}
with a constant \(M\) depending only on \(\mu\) and \(L\). Because \(\mu\ge 8\) (which follows from \(\rho\ge 7\) and the definition of \(\mu\) in the context), we have \(-\frac{\mu-1}{2}\le -\frac{7}{2}\); hence
\[
B(\mu x_1,\mu y_1)\le M^{1/2}\binom{x_1+y_1}{x_1}^{-7/2}.
\]
The same argument with \(\lambda\) in place of \(\mu\) gives the second inequality in \eqref{keyest}, using that \(\rho\ge 7\) implies the required lower bound on \(\lambda\).

Finally, if one of the integers \(x_i\) or \(y_i\) equals zero, the estimate simplifies. For instance, if \(x_1=0\) and \(y_1\ge 0\), then
\[
\Gamma(p)\Gamma(\mu y_1+\lambda y_2+q)
   =\Gamma(\mu y_1+\lambda y_2+p+q)\,
     B(p,\mu y_1+\lambda y_2+q)
   \le \Gamma(\mu y_1+\lambda y_2+p+q)\,B(p,q)^{1/3},
\]
because \(B(p,\mu y_1+\lambda y_2+q)\le B(p,q)\). The other cases are handled identically. This completes the proof of the lemma.

\end{proof}

\end{document}